\documentclass[10pt,leqno,oneside]{extarticle}

\usepackage{mathptmx} % CHANGES THE FONT - LETTERS ARE THICKER AND DARKER.

\usepackage{titlesec}
\titleformat*{\section}{\large\bfseries} % INDICATES THE FONT SIZE OF THE SECTIONS. 
\titleformat*{\subsection}{\large\bfseries} % INDICATES THE FONT SIZE OF THE SUBSECTIONS. 

\usepackage{cite}

\usepackage[small]{caption}
\usepackage{amsmath}
\usepackage{amssymb}
\usepackage{amsfonts}
\usepackage{setspace}
\usepackage{latexsym}
\usepackage{mathrsfs}
\usepackage{epsfig}
\usepackage{amsthm}
\usepackage{yfonts}
\usepackage{graphicx}
\usepackage{color}

\usepackage{MnSymbol}

\usepackage{enumitem}
\setitemize{wide}

\newcommand{\n}{\noindent}
\newcommand{\be}{\begin{equation}}
\newcommand{\ee}{\end{equation}}
\newcommand{\ben}{\begin{displaymath}}

\newcommand{\een}{\end{displaymath}}

\newcommand{\vs}{\vspace{0.2cm}}

\newtheorem{Proposition}{Proposition}
\newtheorem{Theorem}{Theorem}

\newtheorem{Lemma}{Lemma}

\newtheorem{Corollary}{Corollary}

\usepackage{fancyhdr}
\AtBeginDocument{\thispagestyle{plain}}
\fancypagestyle{plain}{
\fancyhead{}
\fancyfoot{}
\fancyhead[LE,RO]{}
\fancyhead[LO,RE]{}
\fancyfoot[R]{\thepage}

\headsep = 20pt
\addtolength{\headwidth}{.2cm}
}

\begin{document}

\thispagestyle{plain}

\n {\huge On the shape of rotating black-holes}\\

\n {\sc Martin Reiris}

\n {\small martin@aei.mpg.de}\\

\n {\sc Maria Eugenia Gabach Clement}

\n {\small megabach@gmail.com}\\

\n \textsc{Max Planck Institute f\"ur Gravitationasphysik \\ Golm - Germany}\\

\n \begin{minipage}[l]{11cm}\begin{spacing}{.9}{\small
We give a thorough description of the shape of rotating axisymmetric stable black-hole (apparent) horizons applicable in dynamical or stationary regimes. It is found that rotation manifests in the widening of their central regions ({\it rotational thickening}), limits their global shapes to the extent that stable holes of a given area $A$ and angular momentum $J\neq 0$ form a precompact family ({\it rotational stabilization}) and enforces their whole geometry to be close to the extreme-Kerr horizon geometry at almost maximal rotational speed ({\it enforced shaping}). The results, which are based on the stability inequality, depend only on $A$ and $J$. In particular they are entirely independent of the surrounding geometry of the space-time and of the presence of matter satisfying the strong energy condition. A complete set of relations between $A$, $J$, the length $L$ of the meridians and the length $R$ of the greatest axisymmetric circle, is given. We also provide concrete estimations for the distance between the geometry of horizons and that of the extreme Kerr, in terms only of $A$ and $J$. Besides its own interest, the work has applications to the Hoop conjecture as formulated by Gibbons in terms of the Birkhoff invariant, to the Bekenstein-Hod entropy bounds and to the study of the compactness of classes of stationary black-hole space-times.}
\end{spacing}\vspace{.3cm}{\sc PACS}: 04.70.Bw, 02.40.-k
\end{minipage}

\vs
\section{Introduction.}

Apparent horizons have been used successfully since decades as the localization of the event horizon along the time evolution \cite{MR0424186}.
In the last years however, they have acquired and even bigger mathematical relevance by the finding that they are stable in a very precise sense \cite{Andersson:2007fh}, \cite{Andersson:2007gy}. 
Based in these new developments we give here a thorough description of the shape of rotating ($J\neq 0$) stable horizons of axisymmetric space-times, only in terms of their area $A$ and their angular momentum $J$. 

The remarkable fact that there are strict constraints on the geometry of axisymmetric apparent horizons arising 
merely from $A$ and $J$ is unique to 3+1 dimensions and differs drastically from what occurs even in 4+1 dimensions 
where extraordinary new phenomena seem to emerge \cite{Lehner:2011wc}. In this article we explore the shape of such horizons to gain insight 
about the shape of realistic black holes in our universe.

Celestial bodies tend to be spherical due to gravity. It is expected that whenever enough and slowly rotating mass is gathered close enough 
together, the resultant gravity will pull equally in all directions and a spherical shape will result. Thus, stars and planets,
on the whole, are close to spherical. 
But when fast rotating matter condensates, the deviations from sphericity of the final shape could be quite common and not necessarily negligible or small.
The most noticeable of these deformations is a flattening perpendicular to the rotation axis of the spinning objects, resulting in 
configurations that become ever more oblate for increasingly rapid rotation. The largest known rotational flattening of a 
star in our galaxy is present in the star Achernar (the ninth-brightest star in the night), which  is spinning so fast that the ratio of the equatorial radius $R_{\rm e}$ to the polar radius 
$R_{\rm p}$ deviates drastically from one, reaching the outstanding $R_{\rm e}/R_{\rm p} \approx1.56$ (\cite{DomicianodeSouza:2003rq}). This implies a flattening
$f:=1-R_{\rm p}/R_{\rm e} \approx 0.35$. The significance of the deviation is evident when comparing it with the flattening of the Sun $f\approx 5\times 10^{-5}$, the Earth $f\approx 3.35\times 10^{-3}$ or Saturn $f\approx 9.79\times 10^{-2}$. Yet, in all these cases, even in the extreme Achernar, flattening is largely a classical phenomenon associated to rotation and with General Relativity playing no role. For Einstein's theory to be significant, the rotational period $T$ of the object must be of the order of $4\pi GM/c^{3}$, or (in geometrized units) the dimensionless quotient
\ben
\Gamma:=\frac{1}{2\Omega M}
\een
where $\Omega$ is the angular velocity, must be of the order of the unity. Achernar, in particular, has $\Gamma \approx 10^{10}$ and even higher ratios hold for the 
other astrophysical objects mentioned above. In Newtonian mechanics, and for uniformly rotating bodies of constant density, 
the quotient $\Gamma$ is closely related to the quotient 
\begin{equation}
\tilde{\Gamma}=\frac{A}{8\pi |J|}
\end{equation}
which depends only on the mass and the geometry of the physical system. Indeed, for spheres we have 
$\tilde{\Gamma}=\frac{5}{2} \Gamma$ and for cylinders with radius equal to their height we have 
$\tilde{\Gamma}=\Gamma$. The quotient $\tilde{\Gamma}$ is also meaningful for axisymmetric 
black holes and will be used fundamentally all through the article.

In comparison to the examples before where $\Gamma$ is exceedingly large, the situation drammatically changes when 
considering millisecond pulsars. For instance a pulsar with a rotational period of one millisecond and of two solar masses, would have $\Gamma \approx 8.3$. To date, the highest rotating pulsar known is PSR J1748-2446ad, with a period $1.4$ milliseconds, and a mass between one or two solar masses. Even a conservative mass of one and a half solar masses would give $\Gamma \approx 15.5$. Two questions thus naturally arise: Are the shapes of millisecond pulsars affected by their high rotations? Is General Relativity playing any role?.

Let us move now to see what occurs to the Kerr black-hole horizons which are by nature General Relativistic. From now on it will be conceptually advantageous to think of horizons as the surfaces of ``abstracted" celestial bodies possessing a mass $M$, an area $A$, an angular momentum $|J|$ and a rotational velocity $\Omega$, all like most ordinary celestial bodies would have\footnote{The abstraction is so useful indeed that, at times, it can be a bit perplexing}. The metrics of the Kerr horizons carry the expressions
\be\label{KHM}
h=\Sigma\, d\theta^{2} + (2M r)^{2}\, \Sigma^{-1}\sin^{2}\theta\, d\varphi^{2}
\ee
where 
\ben
\Sigma=r^{2}+|J|^{2}M^{-2}\cos^{2}\theta,\qquad \text{and}\qquad\ r=M+\sqrt{M^{2}-|J|^{2}M^{-2}}.
\een
The three most basic measures of ``size'' of an axisymmetric black-hole are its area $A$, the length $R$ of its great circle, that is, the length of the greatest axisymmetric orbit, and the length $L$ of the meridian, which is the distance between the poles, as is described in Figure [F1]. For the Kerr black-holes these parameters are given by
\ben
A=8\pi M r,\qquad R=4\pi M,\qquad \text{and}\qquad L=\int_{0}^{\pi}\sqrt{r^{2}+|J|^{2}M^{-2}\cos^{2}\theta}\, d\theta.
\een

\begin{figure}[h]
\centering
\includegraphics[width=9cm,height=8cm]{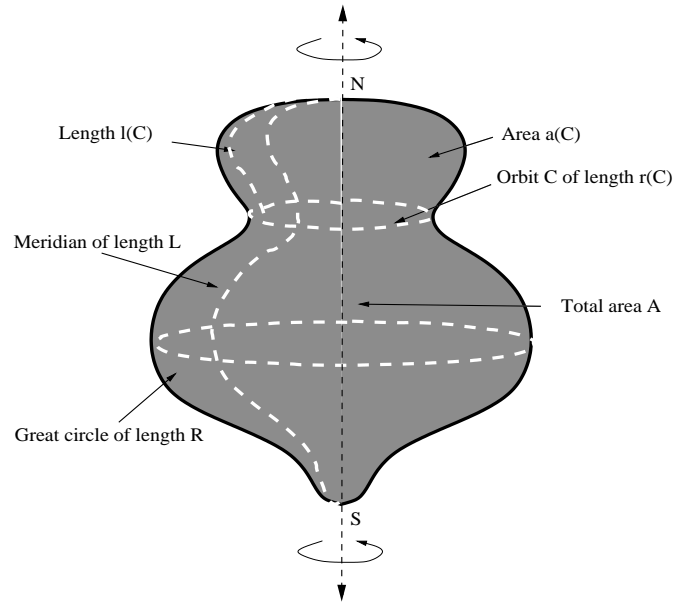}
\caption{Representation of a distorted (dynamical or stationary) axisymmetric horizon and the main geometric parameters.}
\label{Fig1}
\end{figure} 

\n If we fix the mass but increase the rotation from $|J|=0$ all the way until the greatest angular momentum a horizon can hold at $|J|=M$, then the length $R$ of the great circle remains constant but the length $L$ of the meridian deacrases monotonically\footnote{As a direct implicit computation of its derivative shows.}. The flattening $\tilde{f}:=1-2L/R$, in particular, passes from $\tilde{f}=0$ when $|J|=0$ to a maximum $\tilde{f} \approx 0.36$ when $|J|=M$ (note that the flattening coefficient $\tilde{f}$ is not the same as $f$). To compare, 
the Achernar star has $\tilde{f}=0.17$. As expected then, the more the black-holes rotate the more oblate they become.  
Observe that, as $|J|$ varies from $|J|=0$ to $|J|=M$, the quantities $\Gamma=1/2\Omega M$ and $\tilde{\Gamma}=A/8\pi |J|$ for the Kerr-horizons, vary from $\Gamma=\tilde{\Gamma}=\infty$ all the way down to $\Gamma=\tilde{\Gamma}=1$. In particular for the extreme horizon, which is the most oblate one, with $\tilde{f}\approx 0.36$, we have $1/2\Omega M=A/8\pi|J|=1$. Then, although the Kerr horizons are by nature General Relativistic, their rotational flattening is markedly manifest only when $1/2\Omega M\approx A/8\pi |J|\approx 1$.      

It is worth mentioning that none of the rotating Kerr-horizons (i.e. when $|J|\neq 0$) are exactly metrical spheroids and their oblate shapes are not so simple to visualize. To get a better graphical understanding one could isometrically embed them into Euclidean space. This can be done for small values of $|J|$, obtaining then nice oblate spheroidal shapes \cite{Gibbons:2009qe}, but there is a maximum value of $|J|$ (less than $M$) after which isometric embeddings into Euclidean space are no more available
\footnote{To roughly see that such a maximum must exist observe that when $|J|=M$, namely for the fastest rotating black-hole, we have $A=16\pi M^{2}$ and $R=4\pi M$, in particular the areas of the discs $D_{N}$ and $D_{S}$ enclosed by the great circle are both equal to $8\pi M^{2}$. If an isometric embedding exists then the great circle would map into a circle $C$ in Euclidean space and of radius $2M$, but then the discs $D_{N}$, $D_{S}$ would both have to map into the flat disc filling 
$C$, because this is the only disc with boundary $C$ having area $4\pi M^{2}$. All this is a manifestation of the fact that for $|J|$ high the Gaussian curvature of the horizon becomes rather negative near the two poles.}. 
A detailed discussion of these issues is presented in \cite{Gibbons:2009qe} including an analysis of isometric embeddings of horizons into the hyperbolic space.   
For reference, a convenient way to depict axisymmetric holes is the following. For every rotational orbit $C$ let $a(C)$ be the area of the disc enclosed by $C$ and containing the north pole $N$, and let $r(C)$ be the length of $C$. Then on a $(r,a)$-grid, graph $r(a)$, shift it upwards by $A/2$ and flip it around the $\{r\}$-axis (to have the north up). The result is the representation of the black-hole. In the Figure \ref{Fig2} we show the corresponding graphs of the Schwarzschild and extreme-Kerr holes of the same mass (equal to $1/2$). For Schwarzschild in particular the graph is a semicircle. The flattening due to rotation is then evident. We will use this type of representation again in Figure \ref{PON}.    

\begin{figure}[h]
\centering
\includegraphics[width=9cm,height=5cm]{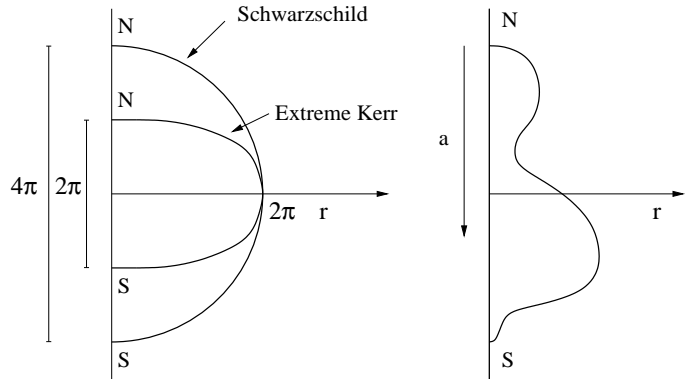}
\caption{Left: The visualization of the Schwarzschild and extreme Kerr black holes in the (r,a)-grid. Right: The representation of the geometry in Figure \ref{Fig1}.}
\label{Fig2}
\end{figure} 

We will investigate flattening and other effects that rotation causes over the shape of black-holes, and we will do so, as we said, only in terms of $A$ and $J$. The reason why it is useful in axisymmetry to control the geometry of horizons in terms of $A$ and $J$ can only be exemplified as follows. Suppose that a single compact body (part of an axisymmetric system) evolves in such a way that at a certain time-slice $\{t_{0}\}$ it is surrounded by trapped surfaces signaling the beginning of gravitational collapse and the emergence of a black hole. At the slice $\{t_{0}\}$ and at any other subsequent slice $\{t\}$, the apparent horizon $H_{t}$ is located at the boundary of the trapped region \cite{MR0424186}, \cite{Andersson:2007gy}. As the material body sinks deep inside the hole the outside region of the apparent horizons stays empty and their angular momentum is conserved, i.e. $J(H_{t})=J$. Moreover, at every time slice $\{t\}$ the universal inequality $8\pi |J|\leq A(H_{t})$ holds \cite{2011PhRvD..84l1503J},\cite{2011PhRvL.107e1101D},\cite{Hennig:2008yw} and we also expect the validity of the Penrose inequality $A(H_{t}) \leq 8\pi M^{2}$, where $M$ is the ADM-mass which is also conserved. Thus, in this scenario we have $J(H_{t})=J$, $8\pi |J|\leq A(H_{t})$ and we expect to have $A(H_{t})\leq 8\pi M^{2}$. Hence, every quantity or property of stable horizons that is proved to be controlled only by the area $A$ and the angular momentum $J$, will be also controlled on the apparent horizons in the process of gravitational collapse. 

One of the first attempts to give information about the shape of black holes goes back to the Hoop conjecture, formulated  by 
Thorne \cite{1818} in 1972. It reads ``{\it Horizons form when and only when a mass $m$ gets compacted onto a region whose 
circumference in every direction is less than or equal to $4\pi M$}". According to this conjecture, the circumference around the 
region must be bounded in every direction, and hence, a thin but long body of given mass would 
not necessarily  evolve to form a horizon. Unfortunately, the impreciseness of Thorne's statement had made this 
heuristic conjecture difficult to state, approach and ultimately, to prove. In this article we assume the presence of a black hole and investigate its geometric properties. In this sense necessary conditions for the formation of black-holes are presented. Particularly, we will validate the picture of the (reciprocal) Hoop conjecture as formulated by Gibbons \cite{Gibbons:2009xm}. This is done in Proposition \ref{CORHO}.

Well defined, intrinsic and useful measures of shape are important in the study of the geometry of black hole horizons. To define them one possibility is to  use a background, well known configuration, to compare with. For rotating black holes, the extreme Kerr black-hole plays a key role, and will be used therefore as the reference metric. In this regard in Theorem \ref{Thmpoint} we are able to estimate the ``distance'' from a given horizon to the extreme Kerr horizon (of the same $J$) only in terms of $A$ and $J$.
One can also red consider global quantities like $R,L$ or $A$ or one can construct dimension-less coefficients, like the flatness coefficient $\tilde{f}=1-2L/R$ mentioned before, that give an intrinsic notion of deformation. Gibbons \cite{Gibbons:2009xm}, \cite{Gibbons:2012ac} for instance, studies the length of the shortest non 
trivial closed geodesic $\ell$ and the Birkhoff's invariant $\beta$. To demonstrate their usefulness he proves that if the surface admits an 
antipodal isometry and that the Penrose inequality holds, then $\ell\leq\sqrt{\pi A}$ and $\ell\leq 4\pi M$. He conjectures that these inequalities hold in the general case, without antipodal symmetry. In Proposition \ref{CORHO} we come very close to proving it as we will get $\ell \leq \beta \leq 2\sqrt{\pi A}$. 
We present many geometric relations of this kind between $R,L, A$ and $J$ which are resumed and discussed in Theorems \ref{ThmRL} and \ref{ThmRL2} and in Proposition \ref{CORHO}. 

Outermost marginally trapped surfaces (MOTS), of which apparent horizons are an instance, are those for which the outgoing null expansion is zero. Stable MOTS are those which can be deformed outwards while keeping the outgoing null-expansion non-negative (to first order) \cite{Andersson:2007gy}. All the results in this article are stated for stable MOTS. To have a flexible terminology we will refer them from now on simply as stable ``horizons'', ``holes'' or ``black-holes''. 

At first sight the stability property seems to be too simple to have any relevant consequence. But indeed and contrary to this perception the stability is crucial and plays a central role in many features of black-holes. It will be also the main tool to be used here. For this reason 
let us give now a glimpse of the main elements of stability in the axisymmetric setup. 
For an axisymmetric and stable black-hole $H$ in a space-time with matter satisfying the strong energy condition, the stability implies the inequality 
\be\label{INSTI1}
\int_{H} \bigg(|\nabla \alpha|^{2} +\kappa \alpha^{2}\bigg)\, dA\geq \int_{H} \frac{|S|^{2}}{2}\alpha^{2}\, dA  
\ee
for any axisymmetric function $\alpha$ on $H$ \cite{2011PhRvD..84l1503J}. Here $\kappa$ is the Gaussian curvature of $H$ with its induced two-metric $h$ and $S$ is the (intrinsic) Hajic\`ek one-form which is defined by $S(X)=-<k,D_{X} l>/2$ where $l$ and $k$ are outgoing and ingoing future null vectors respectively, normalized to have $<k,l>=-2$ but otherwise arbitrary ($D$ is the covariant derivative of the space-time). In terms of $S$, the Komar angular momentum of $H$ is just
\be\label{JDEF}
J(H)=\frac{1}{8\pi} \int_{H} S(\xi)\, dA
\ee 
where $\xi$ is the rotational Killing field. One can use axisymmetry to further simplify (\ref{INSTI1}). We explain how this is done in what follows. Over any axisymmetric sphere there are unique coordinates $(\theta,\varphi)$, called {\it areal 
coordinates}, on which the metric takes the form
\be\label{METRIC}
h=\bigg(\frac{A}{4\pi}\bigg)^{2}e^{\ \displaystyle-\sigma(\theta)}d\theta^{2}+e^{\ \displaystyle \sigma(\theta)}\sin^{2}\theta d\varphi^{2}
\ee
and where $\partial_{\varphi}=\xi$ is, manifestly, the rotational Killing field over $H$. Regularity at the poles implies
$\sigma(0)=\sigma(\pi)=\ln (A/4\pi)$. The area element is $dA=\frac{A}{4\pi}\sin\theta d\theta d\varphi$
and is thus a multiple of the area element of the unit two-sphere. Then define a rotational potential $\omega=\omega(\theta)$ by imposing 
\begin{align*}
& \frac{d\omega}{d\theta}=\frac{A}{2\pi}\sin\theta S(\partial_{\varphi}),\\
& \omega(0)=-\omega(\pi).
\end{align*}
A direct computation using (\ref{JDEF}) then gives $J=(\omega(\pi)-\omega(0))/8=\omega(\pi)/4$. In terms of the coordinates $(\theta,\varphi)$ and $\omega$, the inequality (\ref{INSTI1}) results into 
\be\label{FINE}
\int_{H}\big( |\nabla \alpha|^{2} + \kappa \alpha^{2} \big)\, \sin\theta d\theta d\varphi \geq  \bigg(\frac{2\pi}{A}\bigg)^{2}\int_{H} \frac{|\nabla \omega|^{2}}{\eta\sin\theta}\alpha^{2}\, d\theta d\varphi
\ee
which is valid for any axisymmetric function $\alpha$. Note that the integrands are independent of $\varphi$ and that therefore the integral in $\varphi$ can be factored out to a $2\pi$. The inequality is set out of the two arguments $\omega$ and  $\sigma$, and for this reason $(\omega,\sigma)$ will be our data. Many times however we will use 
\ben
\eta:=e^{\sigma}\sin^{2}\theta
\een
instead of $\sigma$, and use the data $(\omega,\eta)$ instead of $(\omega,\sigma)$. Of particular interest is $(\omega_{E},\sigma_{E})$, the data of the extreme Kerr horizon with angular momentum $J$, which plays the role of a background data and has the expression 
\ben
\sigma_{E}=\ln \frac{4|J|}{1+\cos^{2}\theta},\qquad \omega_{E}=-\frac{8J\cos\theta}{1+\cos^{2}\theta}.
\een
All the results in this article are based on different uses of the fundamental inequality (\ref{FINE}). The difficulty in each case resides in how to chose the trial functions $\alpha$ to get the desired information over $(\omega, \sigma)$.
Let us illustrate this point with an example that will be important to us many times later. Choosing $\alpha=e^{-\sigma/2}$ in (\ref{FINE}) one obtains \cite{2011PhRvL.107e1101D}
\be\label{I}
A\geq 4\pi e^{\displaystyle (\ \mathcal{M}-8)/8}
\ee
where ${\mathcal{M}}=\mathcal{M}(\omega,\sigma)$ is the functional
\be\label{II}
\mathcal{M}(\omega,\sigma)=\int_{0}^{\pi} \bigg(\sigma'^{2} +4\sigma +\frac{\omega'^{2}}{\eta^{2}}\bigg)\sin\theta d\theta.
\ee    
The crucial fact here is that, regardless of the particular functions $(\omega,\sigma)$ (but with $\omega(\pi)=-\omega(0)=8J$) one has \cite{Acena:2010ws}
\be\label{III}
e^{\displaystyle (\mathcal{M}-8)/8}\geq 2|J|.
\ee
Hence, as shown in \cite{2011PhRvL.107e1101D}, the universal inequality $A\geq 8\pi |J|$ follows by choosing $\alpha=e^{-\sigma/2}$. Equations (\ref{I}), (\ref{II}) and (\ref{III}) will be of great use later. Other choices of $\alpha$ give other kind of information as will be shown during the proofs inside the main text.  

We give now a qualitative overview of our main results. They are discussed in full technical detail in the next Section \ref{MS}. The main results can be resumed in the following three effects due to rotation: (A) {\it Rotational thickening}, (B) {\it Rotational stabilization} and (C) {\it Enforced shaping}.

\vs

(A) {\it Rotational thickening}. In line with the discussion above, the most noticeable effect of rotation is a ``widening'' or ``thickening'' of the bulk of the horizons. The more transparent quantitative estimate supporting this phenomenon is given in 
Theorems \ref{ThmRL} and \ref{ThmRL2}, and states that the length $R$ of the great circle is subject to the lower bound 
\be\label{THICKBH}
\frac{16\pi |J|^{2}}{A}\leq \frac{2|J|}{\delta+\sqrt{\delta^{2}+4}}  \leq \bigg(\frac{R}{2\pi}\bigg)^{2} 
\ee
where
\be
\delta=2\sqrt{\bigg(\frac{A}{8\pi|J|}\bigg)^{2}-1}.
\ee
The meaning of (\ref{THICKBH}) is more evident in black holes with a fixed (non-zero) value of the angular momentum per unit of area, $|J|/A$. Written as $\big(\sqrt{64\pi^{3}|J|/A}\big)\sqrt{|J|}\leq R$, the formula (\ref{THICKBH}) says that the length of the greatest axisymmetric orbit is at least as large as a constant (depending on the ratio angular momentum - area) times the square root of the angular momentum. In simple terms, rotation imposes a minimum (non-zero) value for the length of the greatest circle.

The estimate (\ref{THICKBH}) is somehow elegant but doesn't say whether the greatest circle lies in the ``middle region" of the horizon or ``near the poles", nor does it say anything about the size of other axisymmetric circles. Information about the size of axisymmetric circles in the ``middle regions" can be easily obtained from Theorem \ref{T5}. 
To understand this consider the set of axisymmetric circles $C$ at a distance from the north and south poles greater or equal than one third of the distance between the poles, that is greater or equal than $L/3$. Roughly speaking, the set of such circles ``is the central third" of the horizon. Then, the length $r(C)$ of any such circle is greater than $D(\delta)\sqrt{|J|}$ for a certain function $D(\delta)>0$ (which is a function of the ratio $|J|/A$). This fact, which we prove after the statement of Theorem \ref{T5}, generalizes what we obtained for the great circle and gives further support to the idea of ``thickening by rotation".

\vs
We also show that, provided there is an area bound, the length of the great circle, and therfore the length of any axisymmmetric circle, cannot be arbitrarily large. More precisely we prove also in Theorems \ref{ThmRL} and \ref{ThmRL2} the upper bound  
\be
\bigg(\frac{R}{2\pi}\bigg)^{2}\leq 4|J|\frac{\delta+\sqrt{\delta^{2}+4}}{2} \leq \frac{A}{\pi}.
\ee

There are other related manifestations of the influence of rotation in the shape of horizons which are worth mentioning at this point. For instance we prove in Theorem \ref{ThmRL2} the bounds
\be
D(\delta)\leq \frac{R}{L}\leq 2\sqrt{2}\pi.
\ee
These bounds show that stable rotating horizons of a given area $A$ and angular momentum $J\neq 0$, cannot be arbitrarily oblate nor arbitrarily prolate. This phenomenon is depicted in Figure \ref{PON}. More relations between $R, L, A$ and $J$ are given in Theorem \ref{ThmRL2}.

\begin{figure}[h]
\centering
\includegraphics[width=8cm,height=7cm]{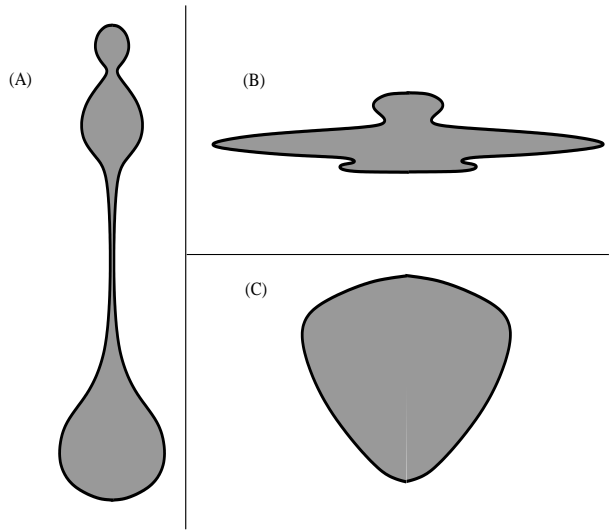}
\caption{When there is a significant rotation $|J|$ in comparison to the area $A$, very prolate or very oblate stable horizons as in (A) or (B)
respectively are forbidden. Instead a shape like (C) would be allowed. The shape of horizons is completely controlled by $\delta=2\sqrt{(A/8\pi|J|)^{2}-1}$.}
\label{PON}
\end{figure}

\vs
(B) {\it Rotational stabilization}. Secondly, we found that rotation stabilizes the shape of stable horizons to such an 
extent that rotating holes of a given area and angular momentum have their entire shapes controlled (and not just their global measures like $R$ or $L$). This is manifest from the pointwise bounds
\be\label{MB2}
\big|\sigma(\theta)-\sigma_{E}(\theta)\big|\leq F(\delta),\qquad \frac{|\omega(\theta)-\omega_{E}(\theta)|}{4|J|}\leq F(\delta)
\ee
for all $\theta\in [0,\pi]$ and for a certain finite function $F(\delta)$, proved in Theorem \ref{Thmpoint}, and which imply the pointwise bounds of the coefficients of the metric $h$ \eqref{METRIC}%which recall has the form
%\ben
%h=\bigg(\frac{A}{4\pi}\bigg)^{2}e^{\displaystyle\ -\sigma(\theta)}d\theta^{2}+e^{\displaystyle\ \sigma(\theta)}\sin^{2}\theta d\varphi^{2}
%\een
%
Still, we are able to prove in Proposition \ref{COMP} the even stronger result that the family of the metrics and potentials of axisymmetric stable horizons of a given area $A$ and angular momentum $J\neq 0$ is precompact (in $C^{0}$). These quantitative facts are specially relevant when applied to apparent horizons in gravitational collapse (as discussed before) revealing a remarkable and unexpected rigidity all along evolution. 

\vs
(C) {\it Enforced shaping}. In third place we found that at very high rotations all the geometry of the horizon tends to that of the extreme Kerr horizon regardless of the presence and type of matter (satisfying the strong energy condition). This claim is also proved in Theorem \ref{Thmpoint}.

\vs
All these results and their applications are discussed in full length in the next sections.

\subsection{Precise statements and further discussions.}\label{MS}

In the sequel we continue using $\delta$ as 
\ben
\delta=2\sqrt{\bigg(\frac{A}{8\pi|J|}\bigg)^{2}-1}.
\een

Our first theorem displays appropriate upper and lower bounds for the length $R$ of the greatest circle. In particular, as commented in $(A)$ above, the lower bound for $R$ in (\ref{Rest1}) shows that rotating black-holes with a given $A$ and $J\neq0$, cannot be arbitrarily ``thin", and the upper bound shows that they cannot be arbitrarily ``thick''.  

\begin{Theorem}\label{ThmRL}
Let $H$ be a stable axisymmetric horizon of area $A$ and angular momentum $J\neq0$. Then the length $R$ of the great circle satisfies 
\be\label{Rest1} 
4|J|\frac{2}{\delta+\sqrt{\delta^{2}+4}}  \leq \left(\frac{R}{2\pi}\right)^{2} \leq 4|J| \frac{\delta+\sqrt{\delta^{2}+4}}{2}.
\ee
\end{Theorem}

\vs
\n These two bounds are sharp, namely they coincide, when $\delta=0$, in $R/2\pi=2\sqrt{|J|}$ which is the value for the extreme Kerr horizon. This is not a coincidence as we will see below that the whole geometry (for a sequence of horizons) converges to that of the extreme Kerr as $\delta\rightarrow 0$. 

Our second theorem displays fundamental relations between the main global geometric parameters $A,J,L$ and $R$ of axisymmetric stable horizons.

\begin{Theorem}\label{ThmRL2}
Let $H$ be a stable axisymmetric horizon of area $A$ and angular momentum $J\neq0$. Then the length $R$ of great circle and the length $L$ of the meridian obey the relations 
\begin{gather}
 \label{Rest2} \frac{(4\pi)^{2} |J|}{\sqrt{4\pi A}}\leq R\leq \sqrt{4\pi A},\\
\label{Rest4} 2\sqrt{|J|} \leq \sqrt{\frac{A}{2\pi}}\leq L,\\   
\label{Rest5} \frac{A}{L^{2}}\leq \frac{R}{L}\leq2\sqrt{2}\pi.
\end{gather}
Moreover there is $D(\delta)>0$ such that
\be\label{Rest3}
D(\delta)\leq \frac{A}{L^{2}}.
\ee
\end{Theorem}

\vs
\n The bounds (\ref{Rest2}) are deduced from (\ref{Rest1}) but are non-sharp at $\delta=0$. The bound $R/2L\leq \sqrt{2} \pi$ expresses that black holes, regardless of the values of $A$ and $J$, cannot be arbitrarily oblate. Note that we would expect the extreme Kerr horizon to be the most flattened black hole, namely we would expect the ratio $R/2L$ to be bounded above by the value of the extreme Kerr horizon, i.e.
$R/2L\leq R_0/2L_0 \sim 0.52 \pi$.  Although non-sharp, the estimation $R/2L\leq \sqrt{2} \pi$ is reasonably good. On the other hand the bound $D(\delta)/2\leq R/2L$ shows that black holes of given $A$ and $J$ cannot be arbitrarily prolate. An expression for $D(\delta)$ can be given explicitly but we will not present it in this article as it is not particularly useful. The existence of $D(\delta)$ will be shown by contradiction. An interesting question is whether one could use the stability inequality (\ref{FINE}), with a suitably chosen probe 
function $\alpha$, to obtain a sharp upper bound on $R/2L$.

Our third theorem displays fundamental relations among the local measures $a,l,r$ of stable rotating holes. Given an axisymmetric orbit $C$, the magnitude $r(C)$ is its length, $l(C)$ is the distance to the north pole $N$ and $a(C)$ is the area of the region enclosed by $C$ and containing the north pole. In the statement below, the parameters $a,l$ and $r$ are defined from the north pole $N$ but of course the same relations hold when they are defined from the south pole $S$.

\begin{Theorem}\label{T5} Let $H$ be a stable axisymmetric horizon. Let $D=A/L^{2}$. Then the following relations hold
\begin{align}
& \label{(a)}\frac{1}{2\pi} a\leq l^{2}\leq \frac{32}{D} a,  \\
&\label{(b)} l \leq \frac{4}{D} r,\  \text{as long as}\  l\leq L/2,\\
&\label{(c)} r^{2}\leq (4\pi e^{4}) a,\ \text{as long as }\ a\leq A/4.
\end{align}
We can use then the inequality $D\geq D(\delta)$ from (\ref{Rest3}) of Theorem \ref{ThmRL2} and (\ref{(a)}), (\ref{(b)}) and (\ref{(c)}) to obtain
\be\label{DD}
a\leq c_{1} l^{2} \leq \frac{c_{2}}{D^{2}(\delta)} r^{2} \leq \frac{c_{3}}{D^{2}(\delta)} a
\ee
as long as $a/A\leq 1/128$ and for certain constants $c_{1},c_{2},c_{3}$. 
\end{Theorem}

\vs
The theorem can be used to obtain varied information. For instance one can extract concrete bounds for the metric coefficient $e^{\sigma}$ around the poles as follows. First note that in (\ref{(c)}), the condition $a(\theta)/A\leq 1/4$ is equivalent to $\theta\leq \pi/3$. This is because $a(\theta)/A=(1-\cos\theta)/2$ and therefore $a/A\leq 1/4$ is equivalent to $\cos\theta\geq 1/2$. Thus, as $r=2\pi e^{\sigma/2}\sin\theta$ we get from (\ref{(c)}) and for $\theta\leq \pi/3$
\ben
4\pi^{2} e^{\sigma}\sin^{2}\theta\leq 2\pi e^{4} A (1-\cos\theta).
\een
But when $\theta\leq \pi/3$ we have $(1-\cos\theta)\leq (1-\cos\theta)(1+\cos\theta)=\sin^{2}\theta$ and therefore $e^{\sigma(\theta)}\leq Ae^{4}/2\pi$ for all $0\leq \theta\leq \pi/3$. By symmetry the same holds for $\theta\in [2\pi/3,\pi]$.

Another application that we commented in the introduction concerns the length of the axisymmetric circles whose distance to the north and the south poles is greater or equal than one third the distance between the poles, that is greater or equal than $L/3$. For any such circle we claimed that $r(C)\geq D(\delta) \sqrt{|J|}$, a relation which gave further support to the idea of ``thickening by rotation". With the help of Theorem \ref{T5} this is proved as follows. Assume, without loss of generality, that the distance from $C$ to the north pole is less or equal than the distance from $C$ to the south pole (i.e. in the notation of Theorem \ref{T5} assume $l\leq L/2$) and then use that $l\geq L/3$ in combination with (\ref{(b)}) and (\ref{Rest4}).

\vs
More general than Theorem \ref{T5}, our fourth theorem shows, as discussed in (B), that the two-metric of the horizon (and therefore its whole geometry) and the rotational potential are completely controlled in $C^{0}$ by $A$ and $J\neq 0$. It also shows that stable holes with $A/8\pi|J|$ close to one must be close to the extreme Kerr horizon.

\begin{Theorem}\label{Thmpoint}
%\begin{enumerate}
%\item[(i)] 
There is $F(\delta)<\infty$ such that for any stable axisymmetric horizon with angular momentum $J\neq 0$ we have 
\be
\label{ET1}\big|\sigma(\theta)-\sigma_{E}(\theta)\big|\leq F(\delta)\qquad \text{and}\qquad \frac{\big|\omega(\theta)-\omega_{E}(\theta)\big|}{4|J|}\leq F(\delta)
\ee
for any $\theta\in [0,\pi]$.
%\item[(ii)]  
Moreover, for any angle $0<\theta_{1}<\pi/2$ and $\epsilon>0$ there is $\bar{\delta}(\theta_{1},\epsilon)$ such that for 
any stable horizon with $\delta <\bar{\delta}$ we have
\be
\label{ET3}\max_{\theta\in [\theta_{1},\pi-\theta_{1}]}\bigg\{\ \big|\sigma(\theta)-\sigma_{E}(\theta)\big|\ \bigg\} \leq \epsilon\qquad \text{and} \qquad \max_{\theta\in [\theta_{1},\pi-\theta_{1}]}\bigg\{\ \frac{\big|\omega(\theta)-\omega_{E}(\theta)\big|}{4|J|}\ \bigg\}\ \leq \epsilon.
\ee
%\end{enumerate}

\end{Theorem}

\vs
The proof of Theorem \ref{Thmpoint} makes use of the following Theorem \ref{Thmw} which is interesting in itself. The Theorem \ref{Thmw} is stated in the variables $(\omega, \eta)$ instead of $(\omega,\sigma)$ and it will be also convenient to think the datum $(\omega,\eta)$ as a path in the hyperbolic plane $\mathbb{H}^{2}=\{(\omega,\eta), \eta>0\}$ provided with the hyperbolic distance
\be\label{HYPM}
d_{\mathbb{H}^{2}}((\omega_{1},\eta_{1}),(\omega_{2},\eta_{2}))={\rm Arch} \bigg[1+\frac{(\omega_{1}-\omega_{2})^{2}+(\eta_{1}-\eta_{2})^{2}}{2\eta_{1}\eta_{2}}\bigg].
\ee
The reason for this is that the functional ${\mathcal{M}}$, on which the Theorem \ref{Thmw} is based, is up to a boundary term the energy of the paths $(\omega,\eta)$ in the hyperbolic plane and such energy functional is easily analyzable. 

\begin{Theorem}\label{Thmw}
Let $H$ be a stable axisymmetric horizon with $J\neq0$. Then
\begin{enumerate}
\item[{\rm (}i\,{\rm )}] The data $(\omega, \eta)=(\omega(\theta),\eta(\theta))$ satisfies
 \be\label{T31}
\frac{(\eta^{2}+\omega^{2}-16|J|^ 2)^{2}}{\eta^{2}}\leq 16|J|^{2}\delta^{2}.
\ee
\item[{\rm (}ii\,{\rm )}] For any two angles $0<\theta_{1}\leq\theta_{2}<\pi$ denote $q_{1}=(\omega_{1},\eta_{1})=(\omega(\theta_{1}),\eta(\theta_{1}))$, $q_{2}=(\omega_{2},\eta_{2})=(\omega(\theta_{1}),\eta(\theta_{2}))$ and $d_{12}=d_{\mathbb{H}^{2}}(q_{1},q_{2})$. Then
\be\label{T4-2}
\bigg|d_{12}-2\ln\bigg[\frac{\tan \theta_{2}/2}{\tan \theta_{1}/2}\bigg]\bigg|^{2} \leq 12 \bigg( \ln \bigg[\frac{\tan \theta_{2}/2}{\tan \theta_{1}/2}\bigg]\bigg)
\ {\rm Arch}\ (1+\delta^{2}).
\ee

\end{enumerate}

\end{Theorem}

\vs
\n Inequality (\ref{T31}) says that the graph of the curve $(\omega,\eta)(\theta)$ in the hyperbolic plane 
$\mathbb{H}^{2}$ lies between the arcs of two circles of centers $(0,2|J|\delta)$ and $(0,-2|J|\delta)$ respectively and both of radius $2|J|\sqrt{4+\delta^{2}}$ (see Figure \ref{Fig-CAGG}). To see this simply observe that (\ref{T31}) implies
\ben
(\eta- 2|J|\delta)^{2}+\omega^{2}\leq 4|J|^{2}(4+\delta^{2})\qquad\text{and}\qquad 
4|J|^{2}(4+\delta^{2})\leq\, (\eta + 2|J|\delta)^{2} +\omega^{2}.
\een
\begin{figure}[h]
\centering
\includegraphics[width=9cm,height=4.5cm]{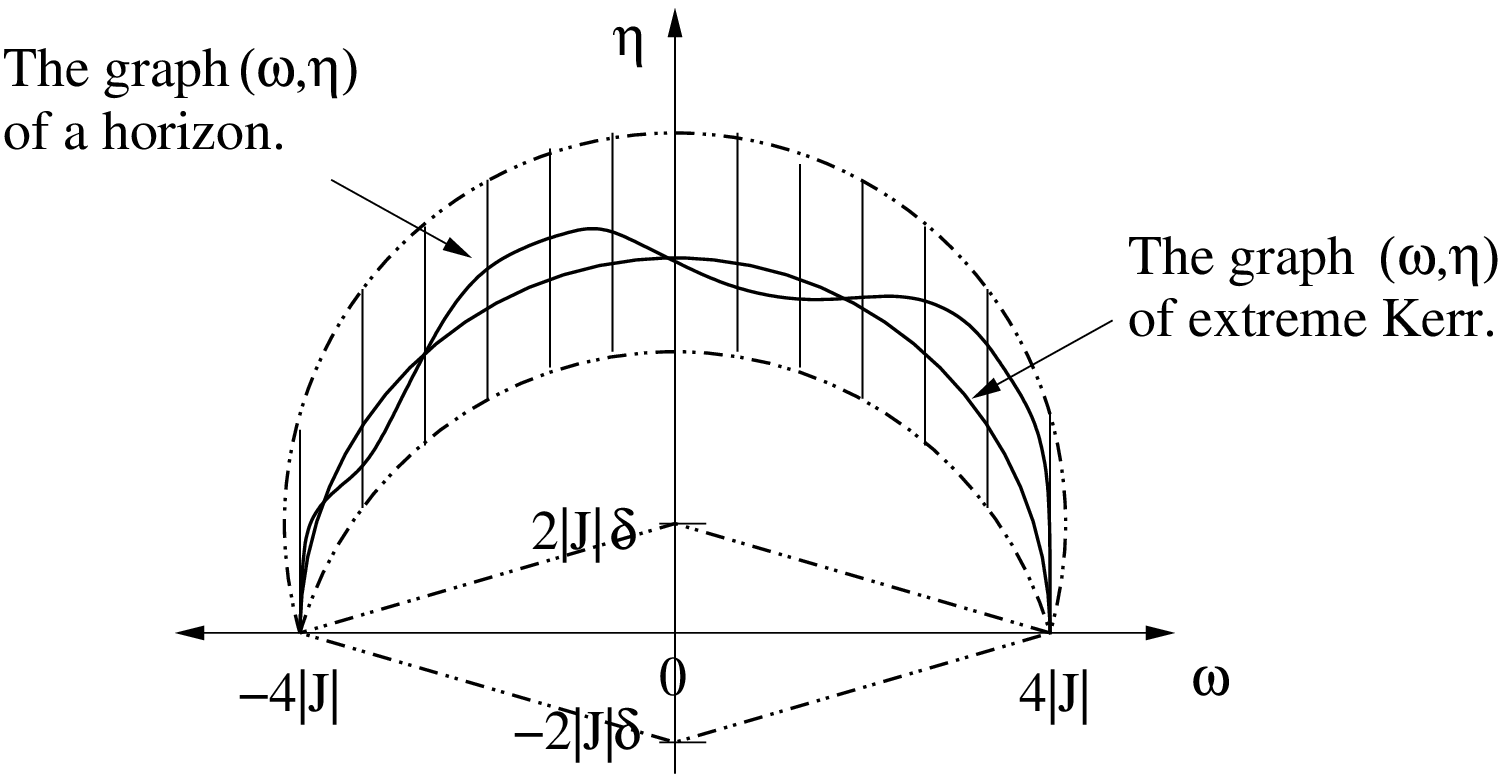}
\caption{}
\label{Fig-CAGG}
\end{figure} 
It is apparent from this that as $A\downarrow 8\pi |J|$ (with $|J|$ fixed), that is, as $\delta\downarrow 0$ and the centers of the circles approach each other, the graph of 
$(\omega,\eta)(\theta)$ gets closer and closer to the unit semicircle which is the graph of extreme Kerr with angular momentum $J$. However this does not imply that 
$(\omega,\eta)(\theta)$, as a parametrized curve, approaches $(\omega_{E},\eta_{E})(\theta)$ as is claimed in Theorem \ref{Thmpoint}. 
%nor does it imply the other statements of Theorem \ref{Thmpoint} item $({\it i}\,)$ in any way. 
It is interesting to see however what occurs if one uses items ({\it i\,}) and ({\it ii\,}) 
in Theorem \ref{Thmw} when $\delta=0$. As we will see this does not imply exactly that $(\omega,\eta)$ is the data of the extreme-Kerr horizon unless we impose that $(\omega(\pi/2),\eta(\pi/2))=(0,1)$. Assume for simplicity of the calculation that $4|J|=1$ (and therefore $A=2\pi$). From (\ref{T31}) one gets 
\be\label{S1}
\eta^{2}+\omega^{2}=1.
\ee
Denote by $\bar{\theta}$ an angle for which $\omega=0$. Because of (\ref{S1}), we also have at this angle $\eta=1$. Using (\ref{T4-2}) with $\theta=\theta_{1}$ and $\bar{\theta}=\theta_{2}$ we obtain
\ben
2\ln \frac{\tan \frac{\theta_{1}}{2}}{\tan\frac{\theta_{0}}{2}}={\rm Arch}\bigg[ 1+\frac{\omega^{2}+(\eta-1)^{2}}{2\eta}\bigg]={\rm Arch} \frac{1}{\eta}
\een
where to obtain the second inequality we have used (\ref{S1}). One can then solve for $\eta$ and once done that use (\ref{S1}) to solve for $\omega$. The result is
\begin{align}
\eta=\frac{2 (\tan^{2} \frac{\theta}{2})\big/(\tan^{2} \frac{\bar{\theta}}{2})}{1+(\tan^{4} \frac{\theta}{2})\big/(\tan^{4} \frac{\bar{\theta}}{2})
%\frac{\tan^{4}\frac{\theta}{2}}{ \tan^{4} \bar{\theta}/2}
},\qquad\omega = \frac{-1+(\tan^{4} \frac{\theta}{2})\big/(\tan^{4} \frac{\bar{\theta}}{2})}{1+(\tan^{4} \frac{\theta}{2})\big/(\tan^{4} \frac{\bar{\theta}}{2})}.
\end{align}
This reduces to the extreme Kerr horizon geometry only when $\bar{\theta}=\pi/2$.

\subsection{Applications.}

\subsubsection{The Hoop conjecture and entropy bounds.}

The following proposition,which is commented below, makes contact with Thorne's Hoop conjecture. 
\begin{Proposition}\label{CORHO} Let $S$ be a stable, axisymmetric, outermost minimal  surface on a maximal axisymmetric and asymptotically flat initial data possibly with matter satisfying the strong energy condition. Then, the length $L$ of the meridian of $S$ and the length $R$ of the great circle satisfy 
\begin{gather}\label{newHoop}
\frac{2\pi |J|}{M}\leq R \leq 8\pi M,\\
\label{EntropyB} |J|^{2}+\bigg(\frac{A}{8\pi}\bigg)^{2}\leq L^{2}M^{2}
\end{gather}
where $M$ is the ADM-mass.
\end{Proposition}  
\begin{proof}[\bf Proof] To obtain (\ref{newHoop}) use the Riemannian Penrose inequality \cite{MR1908823}
\be\label{RPI}
A\leq 16\pi M^{2}
\ee
in equation (\ref{Rest2}). 

To obtain (\ref{EntropyB}) first use $A\leq 2\pi L^{2}$ to get $A^{2}\leq 2\pi A L^{2}$ and then use (\ref{RPI}) on the r.h.s to arrive at 
\be\label{EBOUND1}
\bigg(\frac{A^{2}}{8\pi}\bigg)^{2} \leq \frac{M^{2} L^{2}}{2}.
\ee
Then, from (\ref{Rest4}) we have $|J|\leq L^{2}/4$ and therefore $|J|^{2}\leq |J| L^{2}/4$. Using $|J|\leq 2M^{2}$ (which comes from combining $|J|\leq A/8\pi$ and then (\ref{RPI})) on the r.h.s we arrive at
\be\label{EBOUND2}
|J|^{2}\leq \frac{L^{2}M^{2}}{2}.
\ee
Summing (\ref{EBOUND1}) and (\ref{EBOUND2}) we deduce (\ref{EntropyB}). 
\end{proof}

In \cite{Gibbons:2009xm}, Gibbons proposed that the Birkhoff invariant $\beta$ (see \cite{Gibbons:2009xm} for a definition of $\beta$) of an apparent horizon must verify $\beta\leq 4\pi M$. The aim of Gibbon's proposal was to materialize in a concrete statement Thorne's heuristic  Hoop conjecture.  Quite remarkably we come very close to proving it at least for outermost minimal spheres. Indeed, for an axisymmetric sphere we have always $\beta\leq R$ and therefore from (\ref{newHoop}) we get $\beta\leq 8\pi M$. Whether $8\pi$ instead of Gibbon's $4\pi$ is the right coefficient for $M$ is not known to us. If one expects the Penrose inequality to hold also for apparent horizons, then the argument before would work the same and one would obtain $R\leq 8\pi M$ as well.  
 
On the other hand the equation (\ref{EntropyB}) has a peculiar motivation. In \cite{Bekenstein:1980jp} Bekenstein suggested an upper bound for the entropy of black holes in terms of its ``mass'' $M$ and its ``radius'' $\mathfrak R$, that should hold to guarantee the validity of the generalized second law of Thermodynamics. Bekenstein's suggestion was later extended by Hod \cite{Hod:2000ju} to include angular momentum. According to them the entropy bound should read 
\be\label{EntropyB2}
|J|^{2}+\bigg(\frac{A}{8\pi}\bigg)^{2}\leq {\mathfrak R}^{2}M^{2}.
\ee 
Although the notion of ``radius" is left ambiguously defined, equation (\ref{EntropyB}) shows that (\ref{EntropyB2}) is exactly satisfied if we choose ${\mathfrak R}=L$, that is, the distance from the north to the south pole. Note that by (\ref{Rest1}) we have $4\pi L^{2}\geq A$, showing that $L$ ``qualifies" as a radius according to the point of view of \cite{Bekenstein:1980jp}. 
 
\subsubsection{Compactness of the family of stable rotating horizons.}
 
A remarkable consequence of the results presented in the previous section is that, in appropriate coordinate systems, the space of stable axisymmetric black holes of area $A$ and angular momentum $J\neq 0$ is precompact in the $C^{0}$ topology. This is another strong manifestation of the control that the area and the (non-zero) angular momentum exert on the whole geometry of stable horizons. 

The axisymmetric metric of a horizon can be written in the form 
\ben
h=dl^{2}+\eta(l) d\varphi^{2}
\een
where $l$ varies in $[0,L]$, $\varphi$ in $[0,2\pi]$ and $\eta$ is as before. Instead of the coordinate $l$ we take $x=l/L$. Then $h=L^{2}dx^{2} + \eta(x)d\varphi^{2}$ and $\eta(x):[0,1]\rightarrow \mathbb{R}$ with $\eta(x)=0$ iff $x=0$ or $x=1$. The compactness of the metrics of stable holes is expressed then as follows.
\begin{Theorem}\label{COMP} Let $\{H_{i}\}$ be a sequence of stable rotating horizons having constant area $A$ and angular momentum $J\neq 0$ and having metrics 
\ben
h_{i}=L_{i}^{2}dx^{2}+\eta_{i}(x)d\varphi^{2}.
\een  
Then, there is a subsequence for which the metrics converge in $C^{0}$ to a limit metric
\ben
\bar{h}=\bar{L}^{2} dx^{2} + \bar{\eta}(x)d\varphi^{2}.
\een 
\end{Theorem}  
\n The Theorem can be proved easily and directly from the proposition below. Note that the proposition also shows that the subsequence can be chosen in such a way that a limit for the rotational potential $\omega_{i}$ can also be extracted. Note too that it uses the coordinate $\theta$ rather than $x$. To prove Theorem \ref{COMP} one must change the coordinates from $\theta$ to $x$. The coordinates $(\theta,\varphi)$, where the metrics are expressed in the form $h_{i}=(A/4\pi) e^{-\sigma_{i}}d\theta +\eta_{i}\theta d\varphi^{2}$, are not appropriate because the sequence of coefficients $e^{-\sigma_{i}(\theta)}$ converges weakly but not necessarily in $C^{0}$ near the poles. The coordinates $(x,\varphi)$ reabsorb this problematic coefficient. 
  
\begin{Proposition}\label{PROC} Let $(\omega_{i}(\theta),\eta_{i}(\theta))$ be the data of a sequence of stable horizons $\{H_{i}\}$ of area $A$ and angular momentum $J\neq 0$. Then, there is a subsequence converging in $C^{0}$ to a datum $(\bar{\omega}(\theta),\bar{\eta}(\theta))$.

\end{Proposition} 
  
\begin{proof}[\bf Proof] By Ascoli-Arzel\`a it is enough to show that the sequences $\{\omega_{i}\}$ and $\{\eta_{i}\}$ are uniformly bounded (i.e. $|\omega^{i}|\leq c(A,J)$ and $|\eta^{i}|\leq c_(A,J)$) and equicontinuous (i.e. for all $\epsilon>0$ there is $\delta>0$ such that for any $\theta_{1},\theta_{2}$ with $|\theta_{2}-\theta_{1}|\leq \delta$ we have $|\omega^{i}(\theta_{2})-\omega^{i}(\theta_{1})|\leq \epsilon$ and $|\eta^{i}(\theta_{2})-\eta^{i}(\theta_{1})|\leq \epsilon$). By Theorem \ref{Thmpoint} the sequences $\{\omega_{i}\}$ and $\{\sigma_{i}\}$ are uniformly bounded. Hence also is the sequence $\{\eta_{i}=e^{\sigma_{i}}\sin^{2}\theta\}$. 
We assume then that $|\omega_{i}|\leq c(A,J)$, $|\sigma_{i}|\leq c(A,J)$ and $|\eta_{i}|\leq c(A,J)$. 

We prove next that the sequences are equicontinuous. Observe that if a sequence of functions $\{f_{i}\}$ satifyies $\int_{0}^{\pi} (f{'}_{i})^{2}d\theta\leq c$ then it is equicontinuous as then we would have $|f_{i}(\theta_{2})-f_{i}(\theta_{1})|\leq c^{1/2}|\theta_{2}-\theta_{1}|^{1/2}$. We will show next that the sequences $\{f_{i}=\omega_{i}\}$ and $\{f_{i}=\eta_{i}\}$ have this property. This will finish the proof.
The bound $A\geq e^{(\mathcal{M}(\omega_{i},\eta_{i})-8)/8}$, the bound $|\sigma_{i}|\leq c(A,J)$ and the definition of $\mathcal{M}$ from (\ref{II}) imply
\ben
\int_{0}^{\pi}\frac{w_{i}{'}^{2}}{\eta_{i}^{2}}\sin^{2}\theta d\theta\leq c_{1}(A,J)\qquad \text{and}\qquad \int_{0}^{\pi} \sigma_{i}'^{2}\sin\theta d\theta\leq c_{2}(A,J)
\een
for certain $c_{1}(A,J)$ and $c_{2}(A,J)$. Then, using $|\sigma_{i}|\leq c$ we compute 
\begin{align*}
\int_{0}^{\pi}\omega_{i}^{2} d\theta \leq e^{2c}\int_{0}^{\pi} \frac{\omega_{i}^{2}}{e^{2c}\sin^{3}\theta} d\theta\leq e^{2c} \int_{0}^{\pi} \frac{\omega_{i}{'}^{2}}{\eta_{i}^{2}} \sin\theta d\theta \leq c_{1}e^{2c}.
\end{align*}
On the other hand the following computation proves that $\int_{0}^{\pi}\eta_{i}{'}^{2}d\theta\leq c_{3}(A,J)$:
\begin{align*}
\int_{0}^{\pi} \eta_{i}^{2}d\theta & = \int_{0}^{\pi} (\sigma_{i}'e^{\sigma}\sin\theta - e^{\sigma_{i}}\cos\theta)^{2}d\theta\leq 2\int_{0}^{\pi} (\sigma_{i}^{2} e^{2\sigma_{i}}\sin^{2}\theta + e^{2\sigma_{i}}\cos^{2}\theta)d\theta\\
& \leq 2e^{2c_{}}\int_{0}^{\pi} \sigma_{i}^{2}\sin\theta d\theta +2\pi e^{2c_{}}\leq 2e^{2c_{}}c_{2}+2\pi e^{2c_{}}:=c_{3}(A,J).\qedhere
\end{align*}
Observe that as $|\sigma_{i}|\leq c$ then for any $\theta\neq 0,\pi$ we have $e^{-c}\sin^{2}\theta\leq \eta^{i}(\theta)\leq e^{c}\sin^{2}\theta$ and therefore $\bar{\eta}(\theta)\neq 0$.
\end{proof} 
 
\section{Proofs of the main results.}\label{PMR}

The proof of the results does not follow the order in which they were stated. The order of proof is the following. First we prove Theorem \ref{Thmw} and then Theorem \ref{T5}. After that we prove the bound (\ref{Rest3}) in Theorem \ref{ThmRL2} which is necessary to prove Theorem \ref{Thmpoint}. Finally we give the proofs of Theorems \ref{ThmRL} and \ref{ThmRL2}. Several auxiliary lemmas and propositions are proved in between the main results.

Before we start let us note that when the space-time metric is scaled by $\lambda^{2}$ the following scalings take place
\begin{gather*}
\sigma\rightarrow \sigma + \ln \lambda,\qquad \omega \rightarrow \lambda^{2}\omega,\\
A\rightarrow \lambda^{2}A,\qquad |J|\rightarrow \lambda^{2}|J|,\qquad R\rightarrow \lambda R,\qquad L\rightarrow \lambda L, \\
r\rightarrow \lambda r,\qquad l\rightarrow \lambda l,\qquad a\rightarrow \lambda^{2} a.
\end{gather*}
One can easily see from these scalings that the statements to be proved are scale invariant. For this reason very often we will assume $|J|=1/4$ which is a scale that simplifies considerably the calculations. The assumption will be recalled when used.

\subsection{Proof of Theorem \ref{Thmw}.}

For the proof of Theorem \ref{Thmw} we will use the following lemma. 
\begin{Lemma}\label{ML} Let $(\omega,\eta)$ be any data with $\sigma(0)=\sigma(\pi)$ and $-\omega(0)=\omega(\pi)=1$. For any two angles $0<\theta_{1}\leq \theta_{2}<\pi$ make 
$(\omega_{1},\eta_{1})=(\omega(\theta_{1}),\eta(\theta_{1}))=q_{1},\ (\omega_{2},\eta_{2})=(\omega(\theta_{2}),\eta(\theta_{2}))=q_{2}$ and
\be\label{DAL}
d_{12}=d_{\mathbb{H}^{2}}(q_{1},q_{2})\qquad \text{and}\qquad \alpha_{12}:=\frac{d_{12}}{2\ln \bigg[\frac{\tan \theta_{2}/2}{\tan \theta_{1}/2}\bigg]}.
\ee
Then we have 
\be\label{MFE}
e^{\displaystyle ({\mathcal M}-8)/4}\geq \bigg(\frac{(\omega_{1}+1)^{2}+\eta_{1}^{2}}{4\eta_{1}}\bigg)\bigg(\frac{(\omega_{2}-1)^{2}+\eta_{2}^{2}}{4\eta_{2}}\bigg)\, e^{\displaystyle d_{12}}\, e^{\displaystyle d_{12}(\alpha_{12}-1)^{2}/2\alpha_{12}}.
\ee
\end{Lemma}

\n \begin{proof}[\bf Proof of Lemma \ref{ML}.] Given any data $(\bar{\omega},\bar{\sigma})(\theta)$ defined over an interval $[\bar{\theta}_{1},\bar{\theta}_{2}]$ let's introduce the convenient notation
\ben
{\mathcal M}_{\bar{\theta}_{1}}^{\bar{\theta}_{2}}\big(\bar{\omega},\bar{\eta}\big)=\int_{\bar{\theta}_{1}}^{\bar{\theta}_{2}}\bigg(\bar{\sigma}'^{2}+4\bar{\sigma} +\frac{\bar{\omega}'^{2}}{\bar{\eta}^{2}}\bigg)\, \sin\theta\, d\theta
\een
where $\bar{\eta}=e^{\bar{\sigma}}\sin\theta$. This expression can be conveniently written (\cite{Acena:2010ws}) 
\be\label{FREL}
{\mathcal M}_{\bar{\theta}_{1}}^{\bar{\theta}_{2}}(\bar{\omega},\bar{\eta})=\int_{\bar{\theta}_{1}}^{\bar{\theta}_{2}} \bigg(\frac{\bar{\eta}'^{2}+\bar{\omega}'^{2}}{\bar{\eta}^{2}}\bigg) \sin\theta\, d\theta -4 \bigg((\bar{\sigma}(\theta)+1)\cos\theta+\ln \tan \frac{\theta}{2}\bigg)\, \bigg|_{\bar{\theta}_{1}}^{\bar{\theta}_{2}}
\ee
and observe that by making the change of variable $t=\ln \tan\theta/2$ we get
\be\label{EPATH}
\int_{\bar{\theta}_{1}}^{\bar{\theta}_{2}} \bigg(\frac{\bar{\eta}'^{2}+\bar{\omega}'^{2}}{\bar{\eta}^{2}}\bigg) \sin\theta\, d\theta=\int_{\bar{t}_{1}}^{\bar{t}_{2}} \bigg(\frac{\bar{\eta}'^{2}+\bar{\omega}'^{2}}{\bar{\eta}^{2}}\bigg)\, d t
\ee
where the derivative inside the integral is with respect to $t$. We note too that the right hand side is the energy of the path $(\bar{\omega},\bar{\eta})$ on the hyperbolic plane.  
For this reason, the formulas (\ref{FREL}) and (\ref{EPATH}) show that the minimum of ${\mathcal M}_{\theta_{1}}^{\theta_{2}}\big(\bar{\omega},\bar{\eta}\big)$ among all the paths $\big\{(\bar{\omega}, \bar{\sigma})\big\}$ defined over $[\theta_{1},\theta_{2}]$ and with boundary values $(\omega_{1},\sigma_{1})$ and $(\omega_{2},\sigma_{2})$, is reached at the only geodesic in the hyperbolic plane joining $(\omega_{1},\eta_{1})$ to $(\omega_{2},\eta_{2})$. More precisely if $\gamma(s)=(\bar{\omega}(s),\bar{\eta}(s))$ is the geodesic parametrized by arc-length $s$ starting at $(\omega_{1},\eta_{1})$ (when $s=0$) and ending at $(\bar{\omega}_{2},\bar{\eta}_{2})$ (when $s=d_{12}$), then the minimizing path is
\ben
(\bar{\omega},\bar{\eta})(t)=\gamma\bigg(\frac{d_{12}(t-t_{1})}{t_{2}-t_{1}}\bigg)=\gamma(2\alpha_{12}(t-t_{1})).
\een  
Note that because of this we have $(\bar{\omega}'^{2}+\bar{\eta}'^{2})/\bar{\eta}^{2}=4\alpha_{12}$. 
In this way if we denote the minimum by $\underline{\mathcal M}_{\theta_{1}}^{\theta_{2}}$ then from (\ref{FREL}) and (\ref{EPATH})
we have 
\ben
\underline{\mathcal M}_{\theta_{1}}^{\theta_{2}}=4\bigg((\alpha_{12}^{2}-1)\ln \tan \frac{\theta}{2}-(\sigma(\theta)+1)\cos\theta\bigg)\, \bigg|_{\theta_{1}}^{\theta_{2}}.
\een
On the other hand the minimum of ${\mathcal M}_{1}^{\theta_{2}}$ among all path $(\bar{\omega},\bar{\eta})$ defined over $[0,\theta_{1}]$ with boundary values 
$(-1,0)$ at $\theta=0$ and $(\omega_{1},\eta_{1})$ at $\theta_{1}$, is reached at the unique geodesic in $\mathbb{H}^{2}$ ``from'' $(-1,0)$ to $(\omega_{1},\eta_{1})$. This requires a bit more effort than the previous case, because strictly speaking $(-1,0)$ ``lies'' at infinity in the hyperbolic plane. Nevertheless, a proof can be given exactly as in \cite{Clement:2012vb} or \cite{Chrusciel:2011iv} and won't be repeated here. Being more concrete if $\gamma(s)$ is such geodesic parametrized by arc-length $s$ then the minimum is reached at 
\ben
(\bar{\omega},\bar{\eta})(t)=\gamma(t-t_{1}).
\een
In this way if we denote the minimum by $\underline{\mathcal M}_{0}^{\theta_{1}}$ then from (\ref{FREL}) and (\ref{EPATH}) we have
\ben
\underline{\mathcal M}_{0}^{\theta_{1}}=-4(\sigma(\theta)+1)\cos\theta\bigg|_{0}^{\theta_{1}}.
\een    
The value of $\sigma(0)$ is calculated from the explicit form of the geodesic $\gamma$ mentioned before. The explicit form of the geodesic is 
\ben
\omega(s)=\frac{ABe^{2s}}{B^{2}e^{2s}+1}-1,\qquad
\eta(s)=\frac{Ae^{s}}{B^{2}e^{2s}+1}
\een
where
\ben
A=\frac{(\omega_{1}+1)^{2}+\eta_{1}^{2}}{\eta_{1}},\qquad \text{and}\qquad B=\frac{\omega_{1}+1}{\eta_{0}}.
\een
If we make $s=\ln \tan \theta/2$ and recall that $\eta=e^{\sigma}\sin^{2}\theta$ we obtain the following expression for $\sigma(0)$
\ben
e^{\displaystyle \sigma(0)}=\frac{(\omega_{1}+1)^{2}+\eta_{1}^{2}}{4\eta_{1}}\frac{1}{\tan^{2}\frac{\theta_{1}}{2}}.
\een
One can proceed in the same way to find the minimum $\underline{\mathcal M}_{\theta_{2}}^{\pi}$ among all path $(\bar{\omega},\bar{\eta})$ defined over $[\theta_{2},\pi]$  with boundary values $(\omega_{2},\sigma_{2})$ at $\theta_{2}$ and $(1,0)$ at $\theta=\pi$. The result is 
\ben
\underline{\mathcal M}_{\theta_{2}}^{\pi}=-4(\sigma(\theta)-1)\cos\theta\, \bigg|_{\theta_{2}}^{\pi}
\een
where
\ben
e^{\displaystyle \sigma(\pi)}=\frac{\eta_{2}^{2}+(\omega_{2}-1)^{2}}{4\eta_{2}}\tan^{2}\frac{\theta_{2}}{2}.
\een
Substituting all the lower bounds obtained so far in the r.h.s of the inequality 
\ben
{\mathcal M}={\mathcal M}_{0}^{\theta_{1}}+{\mathcal M}_{\theta_{1}}^{\theta_{2}}+{\mathcal M}_{\theta_{2}}^{\pi}\geq \underline{\mathcal M}_{0}^{\theta_{1}}+\underline{\mathcal M}_{\theta_{1}}^{\theta_{2}}+\underline{\mathcal M}_{\theta_{2}}^{\pi}
\een
and manipulating the expression we get
\begin{align}
e^{\displaystyle ({\mathcal M}-8)/4}&\geq \bigg[\frac{\tan^{\displaystyle \alpha_{12}^{2}-1}\theta_{2}/2}{\tan^{\displaystyle \alpha_{12}^{2}-1}\theta_{1}/2}\bigg]\ e^{\displaystyle \ \big(\sigma(0)+\sigma(\pi)\big)}\\
\nonumber &=\bigg(\frac{(\omega_{1}+1)^{2}+\eta_{1}^{2}}{4\eta_{1}}\bigg)\bigg(\frac{\eta_{2}^{2}+(\omega_{2}-1)^{2}}{4\eta_{2}}\bigg)\ e^{\ \displaystyle d_{12}(\alpha_{12}^{2}+1)/2\alpha_{12}}\\
\nonumber &=\bigg(\frac{(\omega_{1}+1)^{2}+\eta_{1}^{2}}{4\eta_{1}}\bigg)\bigg(\frac{\eta_{2}^{2}+(\omega_{2}-1)^{2}}{4\eta_{2}}\bigg)\ e^{\displaystyle d_{12}}e^{\displaystyle d_{12}(\alpha_{12}-1)^{2}/2\alpha_{12}}
\end{align}
which is the desired inequality.
\end{proof}

\n \begin{proof}[\bf Proof of Theorem \ref{Thmw}]  The statement of Theorem \ref{Thmw} is scale invariant so it is enough to prove it when $|J|=1/4$.

({\it i}\,) In (\ref{MFE}) choose $\theta_{1}=\theta_{2}=\theta$ (and thus $d_{12}=0$) and use the notiation $(\omega,\eta):=(\omega,\eta)(\theta)$ to obtain
\be\label{FF1}
e^{\displaystyle ({\mathcal M}-8)/4}\geq \frac{\big((\omega+1)^{2}+\eta^{2}\big)\big((\omega-1)^{2}+\eta^{2}\big)}{16\eta^{2}}.
\ee
Then manipulate the r.h.s to obtain
\begin{align}\label{FF2}
\frac{\big((\omega+1)^{2}+\eta^{2}\big)\big((\omega-1)^{2}+\eta^{2}\big)}{16\eta^{2}}&
=\frac{(x^{2}+2\omega+1)(x^{2}-2\omega+1)}{16\eta^{2}}
=\frac{(x^{2}+1)^{2}-4\omega^{2}}{16\eta^{2}}\\
\nonumber &=\frac{(x^{2}-1)^{2}+4\eta^{2}}{16\eta^{2}} =
\frac{(\omega^{2}+\eta^{2}-1)^{2}}{16\eta^{2}}+\frac{1}{4}.
\end{align}
We can use this information in the inequality $A\geq 4\pi e^{({\mathcal M}-8)/8}$ to arrive at
\ben
A\geq 4\pi \sqrt{\frac{(\omega^{2}+\eta^{2}-1)^{2}}{16\eta^{2}}+\frac{1}{4}}
\een
which is (\ref{T31}) (recall $|J|=1/4$). 

({\it ii}\,) We move now to prove inequality (\ref{T4-2}). Denote 
\ben
\hat{q}_{1}:=(\hat{\omega}_{1},\hat{\eta}_{1}):=\frac{(\omega_{1},\eta_{1})}{x_{1}},\ \hat{q}_{2}=(\hat{\omega}_{2},\hat{\eta}_{2}):=\frac{(\omega_{2},\eta_{2})}{x_{2}}\qquad\text{and}\qquad d_{\hat{1}\hat{2}}=d_{\mathbb{H}^{2}}(\hat{q}_{1},\hat{q}_{2})
\een
where $x_{1}:=\sqrt{\omega_{1}^{2}+\eta_{1}^{2}}$ and $x_{2}:=\sqrt{\omega_{2}^{2}+\eta_{2}^{2}}$. Of course the points $\hat{q}_{1}$ and $\hat{q}_{2}$ lie in the unit semicircle in the half-plane $\{(\omega,\eta),\eta>0\}$. 
We start by showing that  
\be\label{DEST}
d_{12}\geq d_{\hat{1}\hat{2}}- 2{\rm Arch}(1+\delta^{2}).
\ee
To obtain this inequality it is enough to prove  
\be\label{DHE}
d_{i\hat{i}}\leq {\rm Arch}(1+\delta^{2}),\ \text{for}\ i=1,2,\ \text{(here } d_{i\hat{i}}=d_{\mathbb{H}^{2}}(q_{i},\hat{q}_{i}) \text{)}
\ee
and then use the triangle inequality $d_{\hat{1}\hat{2}}\leq d_{\hat{1}1}+d_{12}+d_{2\hat{2}}$. To prove (\ref{DHE}) recall first that the formula for the hyperbolic distance is
\ben
d_{i\hat{i}}={\rm Arch}\bigg[1+\underbrace{\frac{(\omega_{i}-\hat{\omega}_{i})^{2}+(\eta_{i}-\hat{\eta}_{i})^{2}}{2\eta_{i}\hat{\eta}_{i}}}_{\text{(I)}}\bigg].
\een
Then use $\eta_{i}=x_{i}\hat{\eta}_{i}$ and $\omega_{i}=x_{i}\hat{\omega}_{i}$ to estimate the under-braced term (I) as
\begin{align*}
\frac{(\omega_{i}-\hat{\omega}_{i})^{2}+(\eta_{i}-\hat{\eta}_{i})^{2}}{2\eta_{i}\hat{\eta}_{i}}&=\frac{(x_{i}-1)^{2}}{2x_{i}}\frac{\eta_{i}^{2}+\omega_{i}^{2}}{\eta_{i}^{2}}=\frac{(x_{i}-1)^{2}x_{i}}{2\eta_{i}^{2}}\\
&=\frac{x_{i}}{2(1+x_{i})^{2}} \frac{(\eta_{i}^{2}+\omega_{i}^{2}-1)}{\eta_{i}^{2}} \leq \frac{x_{i}}{2(1+x_{i})^{2}}\delta^{2}\leq \delta^{2}
\end{align*}
as wished (to get the first inequality ($\leq$) we have used (\ref{T31})). 

Let us see  in the sequel how to show (\ref{T4-2}) from the equation (\ref{DEST}) that we have just proved. Plug (\ref{DEST}) in the factor $e^{d_{12}}$ of (\ref{MFE}) to obtain
\begin{align}\label{TEDIOSA}
\bigg(\frac{A}{4\pi}\bigg)^{2} & e^{\displaystyle 2{\rm Arch} (1+\delta^{2})}\geq \underbrace{\bigg(\frac{(\omega_{1}+1)^{2}+\eta_{1}^{2}}{4\eta_{1}}\bigg)\bigg(
\frac{(\omega_{2}-1)^{2}+\eta_{2}^{2}}{4\eta_{2}}\bigg)\, e^{\displaystyle d_{\hat{1}\hat{2}}} }_{\text{(II)}} e^{\displaystyle \frac{d_{12}(\alpha_{12}-1)^{2}}{2\alpha_{12}}}.
\end{align}
We then show that the under-braced factor (II) 
%in front of $e^{d(\alpha-1)^{2}/2\alpha}$ on the right hand side of the previous equation 
can be estimated from below by $1/4$ (i.e. (II)$\geq 1/4$). To see this note that for points in the unit circle the following formula for the hyperbolic distance holds
\be\label{UCE}
e^{\displaystyle d_{\hat{1}\hat{2}}}=\frac{\hat{\eta}_{1}\hat{\eta}_{2}}{(\hat{\omega}_{1}+1)(-\hat{\omega}_{2}+1)}
\ee
and that with it we can compute
\begin{align*}
\nonumber \bigg(\frac{(\omega_{1}+1)^{2}+\eta_{1}^{2}}{4\eta_{1}}\bigg)&\bigg(\frac{(\omega_{2}-1)^{2}+ \eta_{2}^{2}}{4\eta_{2}}\bigg)\, e^{\displaystyle d_{\hat{1}\hat{2}}}=\\ 
\nonumber &=\frac{\big((x_{1}^{2}+1)/x_{1}+2\hat{\omega}_{1}\big)\big((x_{2}^{2}+1)/x_{2}-2\hat{\omega}_{2}\big)}{16\hat{\eta}_{1}\hat{\eta}_{2}}e^{\displaystyle d_{\hat{1}\hat{2}}}\\
\nonumber &=\frac{\big((x_{1}-1)^{2}/x_{1}+2(\hat{\omega}_{1}+1)\big)\big((x_{2}-1)^{2}/x_{2}+2(-\hat{\omega}_{2}+1)\big)}{16 (\hat{\omega}_{1}+1)(-\hat{\omega}_{2}+1)}\\
&=\frac{1}{4}\bigg(\frac{(x_{1}-1)^{2}}{2x_{1}(\hat{\omega}_{1}+1)}+1\bigg)\bigg(\frac{(x_{2}-1)^{2}}{2x_{2}(1-\hat{\omega}_{2})}+1\bigg)\geq \frac{1}{4} 
\end{align*}
where in the second equality we have used (\ref{UCE}) and where the last inequality follows from the fact that because $\hat{q}_{1}$ and $\hat{q}_{2}$ are in the unit semicircle then $1+\hat{\omega}_{1}>0$ and $1-\hat{\omega}_{2}>0$. Finally using the bound (II)$\geq 1/4$ in (\ref{TEDIOSA}) we obtain  
\ben
d_{12}\frac{(\alpha_{12}-1)^{2}}{2\alpha_{12}}\leq 2 {\rm Arch}\, (1+\delta^{2})+2\ln \frac{A}{2\pi}.
\een
The equation (\ref{T4-2}) follows then from using in this equation: (i) the definition of $\alpha_{12}$ in (\ref{DAL}), (ii) that $1+\delta^{2}/4=(A/2\pi)^{2}$ and (iii) that for any $y>0$ we have ${\rm Arch}\, (1+4y)\geq {\rm Arch}\, (1+y)\geq \ln (1+ y)$.
\end{proof}
\begin{figure}[h]
\centering
\includegraphics[width=6.4cm,height=8cm]{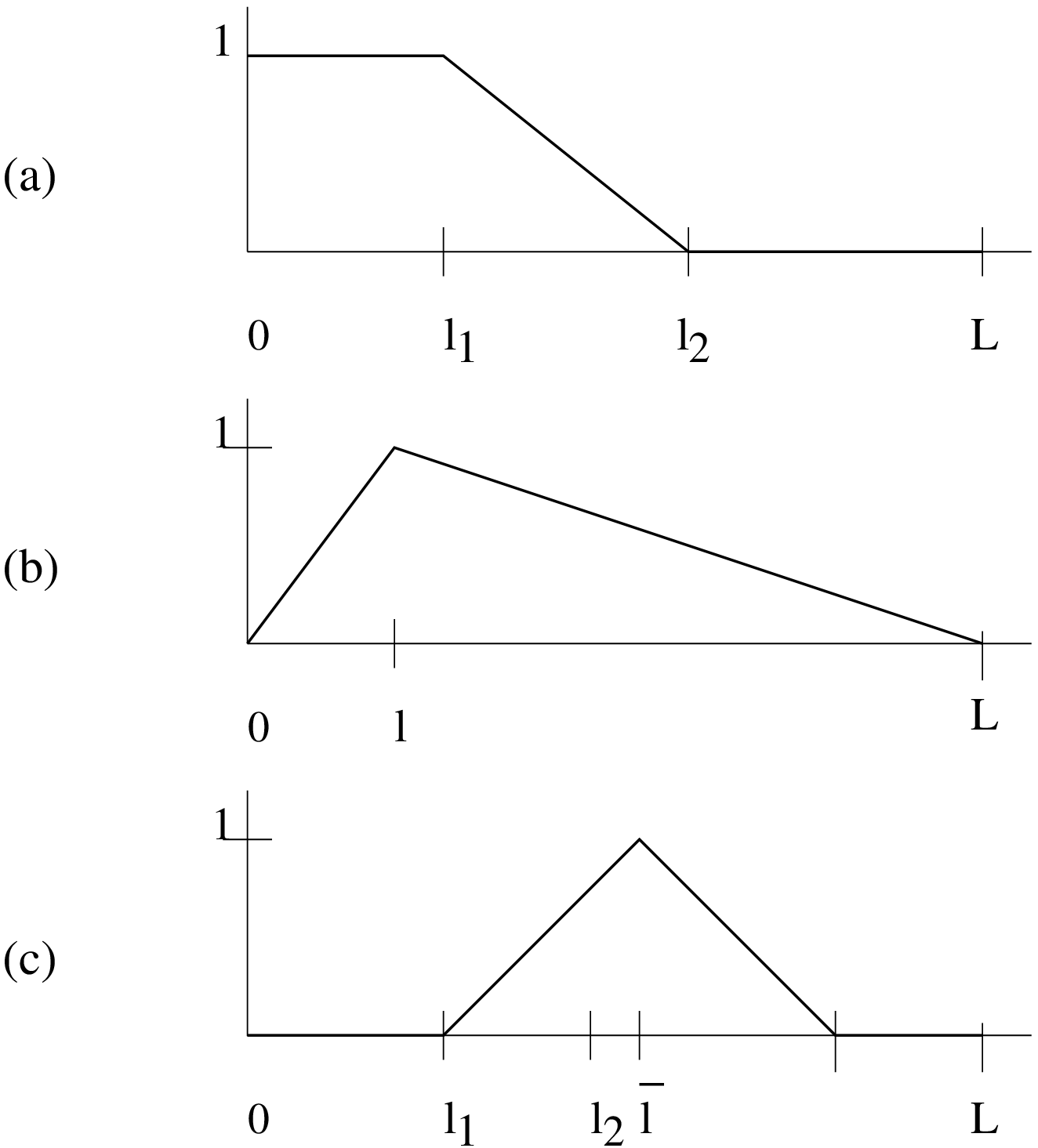}
\caption{}
\label{FGRA}
\end{figure} 
\subsection{Proof of Theorem \ref{T5}.}

\n \begin{proof}[\bf Proof of Theorem \ref{T5}] Let $\alpha=\alpha(l)$ be the linear function in $l\in[l_{1},l_{2}]$ that is one at $l_{1}$ and zero 
at $l_{2}$ (see the graph (a) in Figure \ref{FGRA} between $l_{1}$ and $l_{2}$). Denote by $\Omega_{12}$ the region enclosed by 
the orbits $C_{1}$ and $C_{2}$ at $l=l_{1}$ and $l=l_{2}$ respectively. Let $a_{12}=a_{2}-a_{1}$ be the area of $\Omega_{12}$ and let $l_{12}=l_{2}-l_{1}$. 
We claim that 
\be\label{LTF}
\int_{\Omega_{12}} \big( |\nabla \alpha|^{2}+\kappa\alpha^{2}\big)\, dA= 2\frac{r_{1}}{l_{12}}+r_{1}'-\frac{a_{12}}{l_{12}^{2}}
\ee
where $r_{1}'=(dr/dl)|_{l=l_{1}}$. To see this use that $|\nabla \alpha|^{2}=(\alpha')^{2}=1/l_{12}^{2}$ and that $\kappa=-r''/r$ (where $'=d/dl$) and integrate twice by parts the term $\int_{\Omega_{12}} -\alpha''\alpha\, dA$. 
An important consequence of (\ref{LTF}) is the following. If one takes the trial function $\alpha=\alpha(l)$ equal to one in $[0,l_{1}]$, zero in $[l_{2},L]$ and linear in $[l_{1},l_{2}]$ (see graph (a) in Figure \ref{FGRA}) then   
\ben
\int_{H} \big( |\nabla \alpha|^{2}+\kappa\alpha^{2}\big)\, dA= 2\pi+2\frac{r_{1}}{l_{12}}-\frac{a_{12}}{l_{12}^{2}}.
\een
Any stable horizon has the l.h.s (and therefore the r.h.s) of the previous equation non-negative. Therefore, choosing in it $l_{1}=0$ (therefore $r_{1}=0$) we obtain that for any stable horizon we must have $2\pi l^{2}_{2}\geq a_{2}$ for all $l_{2}$, 
which is the left inequality in (\ref{(a)}).  

Now, if we take a trial function $\alpha=\alpha(l)$, equal to one at $l_{1}$ with $l_{1}\leq L/2$ and linear on 
everyone of the intervals $[0,l_{1}]$, $[l_{1},L]$ (see Figure \ref{FGRA} graph (b)), then, using (\ref{LTF}) over $[0,l_{1}]$ and over $[l_{1},L]$ and summing up one obtains 
\ben
2r_{1}(\frac{1}{l_{1}}+\frac{1}{L-l_{1}})\geq \frac{a_{1}}{l_{1}^{2}}+\frac{A-a_{1}}{(L-l_{1})^{2}}.
\een
Therefore as $l_{1}\leq L/2$ we get
\be\label{T6F1D}
4\frac{r_{1}}{l_{1}}\geq \frac{A}{L^{2}}.
\ee
But as $r_{1}=da/dl|_{l=l_{1}}$ we obtain $4a'\geq (A/L^{2}) l$  (we are making $l_{1}=l$). Integrating we obtain $l^{2}\leq (8 L^{2}/A)a$ as long as $l\leq L/2$. If $l\geq L/2$ then $a(l)\geq a(l/2)\geq (A/(32L^{2}))l^{2}$ which is the right hand side of (\ref{(a)}). We have proved then (\ref{(a)}). 

Formula (\ref{(b)}) is exactly (\ref{T6F1D}). 

Finally we need to prove (\ref{(c)}). We first show that if $H$ is a stable horizon, then for any 
$\theta\leq \pi/3$ we have $\sigma(\theta)\leq c+4$ where $A=4\pi e^{c}$. To see this, let 
$0\leq \theta_{1}<\theta_{2}\leq \pi/3$. Using these angles define a trial function $\alpha$ as
\ben
\alpha=\left\{
\begin{array}{lcl}
e^{-\sigma(\theta_{1})/2} & {\rm if} & \theta\in [0,\theta_{1}],\\
e^{-\sigma(\theta)/2}        & {\rm if} & \theta\in [\theta_{1},\theta_{2}],\\
e^{-\sigma(\theta_{2})/2} & {\rm if} & \theta\in [\theta_{2},\pi]
\end{array}
\right.
\een
With this choice of $\alpha$ we have
\begin{align}\label{SC}
\int_{H}\big( |\nabla \alpha|^{2}+\kappa\alpha^{2}\big)\, dA=2\pi e^{-c}\bigg[ & e^{c-\sigma_{1}}+e^{c-\sigma_{2}}-(\sigma_{2}-c)\cos\theta_{2}+(\sigma_{1}-c)\cos\theta_{1}\\
\nonumber &-\int_{\theta_{1}}^{\theta_{2}}\big( \frac{\sigma'^{2}}{4}+(\sigma-c)\big)\sin\theta d\theta\bigg].
\end{align}
The calculation is straighforward and is explained at the end of the proof. Thus, if $H$ is stable the l.h.s of (\ref{SC}) is non-negative and we must have
\be\label{HHH}
e^{c-\sigma_{1}}+e^{c-\sigma_{2}}\geq (\sigma_{2}-c)\cos\theta_{2}-(\sigma_{1}-c)\cos\theta_{1}+\int_{\theta_{1}}^{\theta_{2}}(\sigma-c)\sin\theta d\theta
\ee
for any $0\leq \theta_{1}<\theta_{2}\leq \pi$. Suppose now that there is $\theta\in (0,\pi/3]$ such that 
$\sigma(\theta)\geq c+4$. Let $\theta_{2}$ be the first angle after the angle zero for which $\sigma$ is equal to $c+4$. 
If for any $\theta$ on $[0,\theta_{2}]$ we have $\sigma(\theta)\geq c$ then, choosing $\theta_{1}=0$ in (\ref{HHH}) we must have
\be\label{CONT}
1+e^{-4}\geq 4\cos\theta_{2}\geq 2
\ee
which is not possible. If there is $\theta$ in $[0,\theta_{2}]$ for which $\sigma(\theta)\leq c$ let $\theta_{1}$ be the first angle before $\theta_{2}$ for which $\sigma$ is equal to $c$. With these choices of $\theta_{1}$ and $\theta_{2}$ in (\ref{HHH}) we obtain again the inequality (\ref{CONT}) which is not possible.

To deduce from this (\ref{(c)}) we note that
\ben
\sin^{2}\theta=4\frac{a(\theta)(A-a(\theta))}{A^{2}}\leq 4\frac{a(\theta)}{A}.
\een
Therefore if $\theta\leq \pi/3$ we obtain
\ben
r^{2}(\theta)=(2\pi)^{2} e^{\sigma(\theta)}\sin^{2}\theta\leq \pi (4\pi e^{c}) e^{4} \sin^{2}\theta \leq 4\pi e^{4} a(\theta)
\een 
from which (\ref{(c)}) follows.

It remains to explain how to perform the calculation (\ref{SC}). We do that in what follows. Denote by $\Omega_{1}$, $\Omega_{12}$ and $\Omega_{2}$ the regions on $H$ corresponding to the $\theta$-intervals $[0,\theta_{1}]$, $[\theta_{1},\theta_{2}]$ and $[\theta_{2},\pi]$ respectively.
For the integration in $\Omega_{1}$, where $\alpha^{2}=e^{-\sigma_{1}}$ use Gauss-Bonnet, $\int_{\Omega_{1}}\kappa dA=2\pi-dr/dl |_{l_{1}}$, and that
\ben
\frac{d r}{d l}\bigg|_{l=l_{1}}=2\pi - 2\pi e^{-c+\sigma_{1}}\big(\frac{\sigma'_{1}}{2}\sin\theta +\cos\theta\big)
\een 
where $\sigma'_{1}=d\sigma/d\theta|_{\theta=\theta_{1}}$. We then obtain
\be\label{AAA}
\int_{\Omega_{1}} \kappa\alpha^{2}\, dA = 2\pi e^{-c}\big( e^{c-\sigma_{1}} - \cos \theta_{1} -\frac{\sigma'_{1}}{2}\sin\theta \bigg).
\ee 
Similarly we have
\be\label{BBB}
\int_{\Omega_{2}} \kappa\alpha^{2}\, dA= 2\pi e^{-c}\bigg(e^{c-\sigma_{2}} + \cos\theta_{2} +\frac{\sigma'_{2}}{2}\sin\theta_{2}
\bigg).
\ee 
For the integration on $\Omega_{12}$ use the expression 
\ben
\kappa=\bigg(\frac{-2\sigma'\cos\theta-\sin\theta \sigma'^{2}+2\sin\theta-(\sin\theta \sigma')'}{2\sin\theta}\bigg)\, e^{\displaystyle -2c+\sigma}
\een
and that $\alpha^{2}=e^{-\sigma}$
to obtain (after integrations by parts)
\begin{align}
\label{OUDOS} \int_{\Omega_{12}} \kappa\alpha^{2}\, dA= & - 2\pi e^{-c}\int_{\theta_{1}}^{\theta_{2}} \frac{\sigma'^{2}}{2}\sin\theta d\theta\\
\nonumber & + 2\pi e^{-c}(\cos\theta_{1}-\cos\theta_{2})\\
\nonumber & + 2\pi e^{-c}(\frac{\sigma'_{1}}{2}\sin\theta_{1} -\frac{\sigma'_{2}}{2}\sin\theta_{2})\\
\nonumber & +2\pi e^{-c}\bigg( (\sigma_{1}-c)\cos\theta_{1} - (\sigma_{2}-c)\cos\theta_{2} - \int_{\theta_{1}}^{\theta_{2}}(\sigma - c) \sin\theta d\theta\bigg).
\end{align} 
Finally add up (\ref{AAA}), (\ref{BBB}) and (\ref{OUDOS}) to deduce (\ref{SC}).
\end{proof}

\subsection{Proof of the bound (\ref{Rest3}) in Theorem \ref{ThmRL2}.}

The proof that there is $D(\delta)$ such that $D\geq D(\delta)$ follows directly from the next two lemmas whose proofs are given immediately after their statements.

\begin{Lemma}\label{L1} Let $H$ be a stable axisymmetric horizon with $4|J|=1$. Let $\Omega_{12}$ be the region on $H$ bounded by two orbits $C_{1}$ and $C_{2}$. Let $l_{12}=l_{2}-l_{1}$ and $a_{12}=a_{2}-a_{1}$ be the distance and the area between them respectively. Then, either $L\leq 5l_{12}$ or
$L\leq 4l_{12}(A/a_{12})^{4}$. Therefore
\ben
D=\frac{A}{L^{2}}\geq \frac{2\pi}{\displaystyle \bigg(\max\bigg\{5l_{12},4l_{12}(A/a_{12})^{4}\bigg\}\bigg)^{2}}.
\een
\end{Lemma} 

\begin{Lemma}\label{L2} Let $H$ be a stable axisymmetric horizon with $4|J|=1$ and area $A$. Then there are orbits $C_{1}$ and $C_{2}$ for which (following the notation of the Lemma \ref{L1})
\begin{align}
\label{CON1}& a_{12}\geq \tilde{a}(A)>0,\\ 
\label{CON2}& l_{12}\leq \tilde{l}(A)<\infty
\end{align}
for certain functions $\tilde{a}(A)$ and $\tilde{l}(A)$.
\end{Lemma}

From these two lemmas the claim (\ref{Rest3}) of Theorem \ref{ThmRL2} is now direct. We state it as a corollary.

\begin{Corollary} Let $H$ be a stable axisymmetic horizon of area $A$ and $J\neq 0$. Then, there is $D(\delta)>0$ such that
\ben
D(\delta)\leq D=\frac{A}{L^{2}}.
\een
\end{Corollary}

\n \begin{proof}[\bf Proof of Lemma \ref{L1}] If $L\leq 5l_{12}$ there is nothing to prove. Assume then that $L> 5l_{12}$ and assume without loss of generality that the middle point between $l_{1}$ and $l_{2}$ (that is $(l_{1}+l_{2})/2$) lies in the interval $[0,L/2]$ (that is $L/2\geq (l_{1}+l_{2})/2$). These two facts imply directly that
\ben
L-l_{1}\geq L/2 \qquad \text{and}\qquad L-l_{1}\geq 3l_{12}
\een
and from them we get 
\ben
\frac{L+l_{1}}{2}\leq \frac{L}{2}+l_{1}\leq L\qquad \text{and}\qquad \frac{L+l_{1}}{2}-l_{2}\geq \frac{l_{1}}{2}+3\frac{l_{2}-l_{1}}{2}\geq \frac{l_{2}-l_{1}}{2}.
\een 
Therefore the interval $[l_{2},(L+l_{1})/2]$ lies inside $[0,L]$ and has a length greater or equal than $l_{12}/2$. Now, for every $\bar{l}\in [l_{2},(L+l_{1})/2]$ consider the trial function $\alpha_{\bar{l}}=\alpha_{\bar{l}}(l)$ 
\begin{equation}
\alpha_{\bar{l}}(l)=\left\{
\begin{array}{lll}
0 & \text{ if } & l\leq l_{1} \text{ or } l\geq \bar{l}+(\bar{l}-l_{1})=2\bar{l}-l_{1},\\
1 & \text{ if } & l=\bar{l},\\
\text{linear} & \text{ when } &l_{1}\leq l\leq \bar{l},\\ 
\text{linear } & \text{ when } & \bar{l}\leq l\leq 2\bar{l}-l_{1}
\end{array}
\right.
\end{equation}
described in Figure \ref{FGRA} graph (c). We use this trial function now and with the help of (\ref{LTF}) (used twice, over $[l_{1},\bar{l}]$ and over $[\bar{l},2\bar{l}-l_{1}]$) we obtain easily
\ben
\int_{H} (|\nabla \alpha_{\bar{l}}|^{2}+\kappa \alpha_{\bar{l}}^{2})\, dA = 4\frac{r(\bar{l})}{\bar{l}-l_{1}} - \frac{a(\bar{l})-a(l_{1})}{(\bar{l}-l_{1})^{2}} - \frac{a(2\bar{l}-l_{1})-a(\bar{l})}{(\bar{l}-l_{1})^{2}}.
\een
In particular, if $H$ is stable then the r.h.s is non-negative and we have
\ben
4r(\bar{l}) \geq \frac{a(\bar{l})-a(l_{1})}{\bar{l}-l_{1}}.
\een
But $r(\bar{l})=d a(\bar{l})/d\bar{l}:=a'(\bar{l})$ and therefore $4a'(\bar{l})/(a(\bar{l})-a(l_{1}))\geq 1/(\bar{l}-l_{1})$. Integrating this inequality for $\bar{l}$ between $l_{2}$ and $(L+l_{1})/2$ we obtain
\ben
\bigg(\frac{a((L+l_{1})/2)-a(l_{1})}{a(l_{2})-a(l_{1})}\bigg)^{4}\geq \frac{(L-l_{1})/2}{l_{2}-l_{1}}.
\een
As $\bar{a}((L-l_{1})/2)\leq A$ and $(L-l_{1})/2\geq L/4$ we deduce
\ben
4l_{12}\bigg(\frac{A}{a_{12}}\bigg)^{4}\geq L
\een
as wished.
\end{proof}

\n \begin{proof}[\bf Proof of Lemma \ref{L2}] In this proof we are assuming that $|J|=1/4$. Take into account therefore that as $(A/2\pi)^{2}=1+\delta^{2}/4$ any function of $A$ can be thought as a function of $\delta$ and vice-versa.

To start, recall that the graph of the data $(\omega,\eta)$ inside the half plane $\{(\omega,\eta)/\eta>0\}$ lies between two arcs of circles passing through $(-1,0)$ and $(1,0)$ but cutting the half-line $\{\eta>0\}$ at the points $(\delta+\sqrt{\delta^{2}+4})/2$ and $(-\delta +\sqrt{\delta^{2}+4}))/2$ respectively (see Figure \ref{Fig-CAG}). Observe too that $(-\delta +\sqrt{\delta^{2}+1}))/2>1/(1+\delta)$. Therefore for any $\tilde{\eta}<1/(1+\delta)$ and angle $\theta$ such that $\eta(\theta)=\tilde{\eta}$ we have either $\omega(\theta)<0$ or 
$\omega(\theta)>0$. For any $\tilde{\eta}<1/(1+\delta)$ define the angles $\theta_{1}=\theta_{1}(\tilde{\eta})$ and $\theta_{2}=\theta_{2}(\tilde{\eta})$ by
\begin{align*}
\theta_{1}=\max \bigg\{\theta/\eta(\theta)=\tilde{\eta}\ \text{and}\ \omega(\theta)<0\bigg\}\qquad\text{and}\qquad
\theta_{2}=\min \bigg\{\theta>\theta_{1}/\eta(\theta)=\tilde{\eta}\ \text{and}\ \omega(\theta)>0\bigg\}.
\end{align*}
With this definition of $\theta_{1}$ and $\theta_{2}$ we clearly have 
\begin{enumerate}
\item[(c0)] $0<\theta_{1}<\theta_{2}<\pi$, and,  
\item[(c1)]  $\eta(\theta_{1})=\eta(\theta_{2})=\tilde{\eta}$, and,
\item[(c2)]  $\eta(\theta)\geq \tilde{\eta}$ when $\theta\in [\theta_{1},\theta_{2}]$, and,
\item[(c3)]  $\omega_{1}=\omega(\theta_{1})<0$ and $\omega_{2}=\omega(\theta_{2})>0$.
\end{enumerate}
\begin{figure}[h]
\centering
\includegraphics[width=6cm,height=5cm]{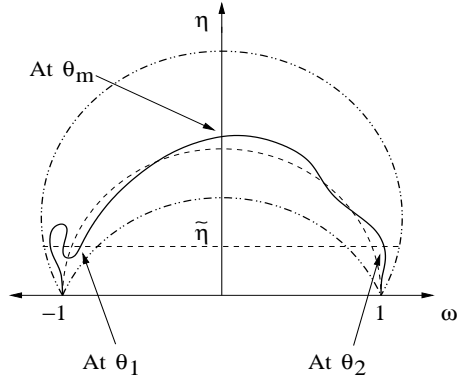}
\caption{The angles $\theta_{1}, \theta_{m}$ and $\theta_{2}$.}
\label{Fig-CAG}
\end{figure} 
Because of (c3) there is $\theta_{m}\in (\theta_{1},\theta_{2})$ such that $\omega(\theta_{m})=0$. Observe that at $\theta_{m}$ we have
\be\label{HW}
\eta(\theta_{m})>\frac{-\delta+\sqrt{\delta^{2}+4}}{2}>\frac{1}{1+\delta}.
\ee
Recall now from the discussion after Theorem \ref{T5} that when $\theta\in [0,\pi/3]\cup [2\pi/3,\pi]$ we have
\ben
\eta(\theta)=e^{\sigma(\theta)}\sin^{2}\theta\leq Ae^{4}\sin^{2}\theta.
\een
From this fact and (\ref{HW}) we deduce that either $\theta_{m}\in (\pi/3,2\pi/3)$ or 
\ben
\frac{1}{1+\delta}<\eta(\theta_{m})\leq Ae^{4}\sin^{2}\theta_{m}.
\een
It follows then that there is $\gamma(\delta)=\gamma(A)\in (0,\pi/2)$ independent of $\tilde{\eta}(<1/(1+\delta)$ such that $\theta_{m}\in [\gamma(A),\pi-\gamma(A)]$. We will use this information below.

In the following denote by $d_{12}$ the hyperbolic distance between $(\omega_{1},\eta_{1})=(\omega_{1},\tilde{\eta})$ and $(\omega_{2},\eta_{2})=(\omega_{2},\tilde{\eta})$. We will use also $\alpha_{12}$ as in (\ref{DAL}). The proof of the Lemma will come from using the inequalities 
%(\ref{I1}), (\ref{I2}) and (\ref{I3}) below
\begin{align}
\label{I1} & d_{12} \geq {\rm Arch}\ \frac{1}{(1+\delta)^{2}\tilde{\eta}^{2}},\\
\label{I2} & \frac{(\alpha_{12}-1)^{2}}{\alpha_{12}}\leq 3\frac{{\rm Arch}\ (1+\delta^{2})}{d_{12}},\\
\label{I3} & \bigg(\min\bigg\{\tan \frac{\theta_{1}}{2},\tan \frac{\pi-\theta_{2}}{2}\bigg\}\bigg)^{2}\leq e^{\displaystyle -d_{12}/2\alpha_{12}}
\end{align}
which are deduced as follows. The inequality (\ref{I1}) follows from
\be\label{PI1}
d_{12}={\rm Arch}\bigg[1+\frac{(\omega_{2}-\omega_{1})^{2}}{2\tilde{\eta}^{2}}\bigg]
\ee
and by noting that
\ben
(\omega_{2}-\omega_{1})^{2}\geq \omega_{1}^{2}+\omega_{2}^{2}\geq 2(\frac{1}{(1+\delta)^{2}}-\tilde{\eta}^{2})
\een
where for the first inequality we use the conditions $\omega_{1}<0$, $\omega_{2}>0$, and for the second we use that for $i=1,2$ we have $\omega_{i}^{2}+\tilde{\eta}^{2}\geq 1/(1+\delta)^{2}$ (the graph of $(\omega,\eta)$ lies outside the disc of center $(0,0)$ and radius $1/(1+\delta)$). The inequality (\ref{I2}) is precisely (\ref{T4-2}) and finally the inequality (\ref{I3}) follows from (\ref{DAL}) and after noting that,  
$\tan (\theta_{2}/2)=1/\tan \big((\pi-\theta_{2})/2\big)$. 

Now, from (\ref{I1}), (\ref{I2}) and (\ref{I3}) we deduce directly the limits
\begin{align*}
\lim_{\tilde{\eta}\rightarrow 0} d_{12}\ = \infty,\qquad \lim_{\tilde{\eta}\rightarrow 0} \alpha_{12} = 1,\qquad\text{and}\qquad
\lim_{\tilde{\eta}\rightarrow 0} \min\bigg\{\tan \frac{\theta_{1}}{2},\tan \frac{\pi-\theta_{2}}{2}\bigg\} = 0.
\end{align*}
Therefore one can chose $\tilde{\eta}=\tilde{\eta}(A)$ such that 
\ben
\min\bigg\{\tan \frac{\theta_{1}}{2},\tan \frac{\pi-\theta_{2}}{2}\bigg\}\leq \tan \frac{\gamma(A)}{4}.
\een
Hence, either $\theta_{1}\leq \gamma(A)/2$ or $\pi-\theta_{2}\leq \gamma(A)/2$. But $\theta_{m}\in [\gamma(A),\pi-\gamma(A)]$ and therefore either $[\gamma(A)/2,\gamma(A)]\subset [\theta_{1},\theta_{2}]$ or $[\pi-\gamma(A),\pi-\gamma(A)/2]\subset [\theta_{1},\theta_{2}]$. We use now this crucial fact to show (\ref{CON1}) and (\ref{CON2}). 

If we denote by $C_{1}$ and $C_{2}$ the axisymmmetric orbits at $\theta_{1}$ and $\theta_{2}$ respectively, then the area $a_{12}$ between them is greater or equal than the area contained either between the orbits with angles $\gamma(A)/2$ and $\gamma(A)$, or between the orbits with angles $\pi-\gamma(A)$ and $\pi-\gamma(A)/2$. In either case  such area is $\tilde{a}(A):=A(\cos \gamma(A)/2 -\cos \gamma(A))/2$. Thus, with this definition of $\tilde{a}(A)$, we have $a_{12}\geq \tilde{a}(A)$ which is (\ref{CON1}).

On the other hand the length $l_{12}$ between $C_{1}$ and $C_{2}$ can be estimated from above by
\ben
l_{12}=\sqrt{\frac{A}{4\pi}}\int_{\theta_{1}}^{\theta_{2}}e^{-\sigma(\theta)/2}d\theta= \sqrt{\frac{A}{4\pi}}\int_{\theta_{1}}^{\theta_{2}}\frac{\sin \theta}{\sqrt{\eta(\theta)}} d\theta \leq \frac{2}{\sqrt{\tilde{\eta}}}\sqrt{\frac{A}{4\pi}}:=\tilde{l}(A)
\een
where we have used $\eta=e^{\sigma}\sin^{2}\theta$ and that, because of ($c_{2}$), between $\theta_{1}$ and $\theta_{2}$ we have $\eta\geq \tilde{\eta}$. With this definition of $\tilde{l}(A)$ we have $l_{12}\leq \tilde{l}(A)$ which is (\ref{CON2}).
\end{proof}

\subsection{Proof of Theorem \ref{Thmpoint}.}

\begin{proof}[\bf Proof of Theorem \ref{Thmpoint}] Again, the statement of Theorem \ref{Thmpoint} is scale invariant and therefore we can assume without loss of generality that $|J|=1/4$. Take into account below that $(A/2\pi)^{2}=1+\delta^{2}/4$ and therefore that any function of $A$ can be thought as a function of $\delta$ and vice-versa.

({\it i}\,) We need to show that there are functions $F_{1}(\delta)$ and $F_{2}(\delta)$ such that for any stable axi-symmetric horizon with data $(\omega,\eta)$ we have $|\sigma-\sigma_{E}|\leq F_{1}(\delta)$ and $|\omega-\omega_{E}|\leq F_{2}(\delta)$. We prove first the later bound and then the former.

{\it The bound $|\omega-\omega_{E}|\leq F_{2}(\delta)$}: We know that the graph of $(\omega,\eta)$ lies inside the region enclosed by the segment $[-1,1]$ on the $\omega$-axis and the arc of a circle of center $(0,\delta/2)$, which starts at $(-1,0)$ and ends at $(1,0)$ (see Figure \ref{Fig-CAGG}). The radius of such circle is $\sqrt{1+\delta^{2}/4}$. Therefore $|\omega|\leq \sqrt{1+\delta^{2}/4}$, which gives the bound
$|\omega-\omega_{E}|\leq \sqrt{1+\delta^{2}/4}+1:=F_{2}(\delta)$ (because $|\omega_{E}|$ is obviously bounded by one (recall $4|J|=1$)).

{\it The bound $|\sigma-\sigma_{E}|\leq F_{1}(\delta)$}: This bound will follow as the result of the two items (``$\bullet$") below. Let $\hat{\theta}\in [0,\pi/2]$ such that $(1-\cos\hat{\theta})/2=1/128$. In the first item ($\bullet$) we prove that $|\sigma(\theta)-\sigma_{E}(\theta)|\leq G_{1}(\delta)$ for certain function $G_{1}(\delta)$ and as long as $\theta\in ([0,\hat{\theta}]\cup [\pi-\hat{\theta},\pi])$. In the second item ($\bullet$) we instead show that $|\sigma(\theta)-\sigma_{E}(\theta)|\leq G_{2}(\delta)$ for certain function $G_{2}(\delta)$ and as long as $\theta\in [\hat{\theta},\pi-\hat{\theta}]$.
Thus, after the two items we will have proved that $|\sigma(\theta)-\sigma_{E}(\theta)|\leq \max\{G_{1}(\delta),G_{2}(\delta)\}:=F_{1}(\delta)$ for any $\theta\in [0,\pi]$ as wished.

$\bullet$ By ({\it ii}\,) in Theorem \ref{T5} we have
\ben
\bar{c}_{1}D^{2}(\delta) a \leq r^{2}\leq \bar{c}_{2} a
\een
for constants $\bar{c}_{1},\bar{c}_{2}$ and as long as $a/A=(1-\cos\theta)/2\leq 1/128$. Recalling that $r^{2}=4\pi e^{\sigma}\sin^{2}\theta$ and that $\sigma_{E}=\ln \big[1/(1+\cos^{2}\theta)\big]$ we get
\ben
\frac{\bar{\bar{c}}_{1}D^{2}(\delta)A(1-\cos\theta)(1+\cos^{2}\theta)}{\sin^{2}\theta}\leq e^{\displaystyle\, \sigma-\sigma_{E}}\leq \frac{\bar{\bar{c}}_{2}A(1-\cos\theta)(1+\cos^{2}\theta)}{\sin^{2}\theta}
\een
for constants $\bar{\bar{c}}_{1}$ and $\bar{\bar{c}}_{2}$. It follows that $|\sigma(\theta)-\sigma_{E}(\theta)|\leq G_{1}(\delta)$ for certain function $G_{1}(\delta)$ and as long as $\theta\in [0,\hat{\theta}]$. By symmetry the similar result applies for $\theta\in [\pi-\hat{\theta},\pi]$. 

$\bullet$ To simplify the notation below we make $(\omega(\theta),\eta(\theta))=(\omega,\eta)$, $(\omega(\hat{\theta}),\eta(\hat{\theta}))=(\hat{\omega},\hat{\eta})$, $\sigma=\sigma(\theta)$ and $\hat{\sigma}=\sigma(\hat{\theta})$. To start we observe that $d((\omega,\eta),(\hat{\omega},\hat{\eta}))\geq 2|\sigma-\hat{\sigma}|$ which is the result of the following computation 
\begin{align*}
d((\omega,\eta),(\hat{\omega},\hat{\eta})))&={\rm Arch}\bigg[1+\frac{(\omega-\hat{\omega})^{2}+(\eta-\hat{\eta})^{2}}{2\eta\hat{\eta}}\bigg]\\
&\geq {\rm Arch}\bigg[ 1+\frac{(\eta-\hat{\eta})^{2}}{2\eta\hat{\eta}}\bigg]={\rm Arch}\bigg[\frac{(\eta/\hat{\eta})^{2}+(\hat{\eta}/\eta)^{2}}{2}\bigg]\\
&={\rm Arch}\big[ \cosh 2(\sigma-\hat{\sigma})\big]=2|\sigma-\hat{\sigma}|.
\end{align*}
We will make use now of (\ref{T4-2}) with $\theta_{1}=\theta$ and $\theta_{2}=\hat{\theta}$, namely
\ben
\bigg[d\big((\omega,\eta),(\hat{\omega},\hat{\eta})\big)-\ln \frac{\tan \hat{\theta}/2}{\tan \theta/2}\bigg]^{2}\leq 4(\ln \frac{\tan\hat{\theta}/2}{\tan\theta/2}){\rm Arch}(1+\delta^{2}).
\een
From this and the inequality $d((\omega,\eta),(\hat{\omega},\hat{\eta}))\geq 2|\sigma-\hat{\sigma}|$ observed before, we get
\ben
2|\sigma-\hat{\sigma}|\leq \ln \frac{\tan \hat{\theta}/2}{\tan \theta/2} +\sqrt{4(\ln \frac{\tan\hat{\theta}/2}{\tan\theta/2}){\rm Arch}(1+\delta^{2})}.
\een
Thus there is $G_{3}(\delta)$ such that for any $\theta\in [\hat{\theta},\pi-\hat{\theta}]$ we have $|\sigma-\hat{\sigma}|\leq G_{3}(\delta)$.
But $|\sigma(\theta)-\sigma_{E}(\theta)|\leq |\sigma(\theta)-\sigma(\hat{\theta})|+|\sigma(\hat{\theta})-\sigma_{E}(\hat{\theta})|+|\sigma_{E}(\hat{\theta})-\sigma_{E}(\hat{\theta})|$ and thus
\ben
|\sigma(\theta)-\sigma_{E}(\theta)|\leq G_{3}(\delta)+G_{1}(\delta)+G_{4}(\hat{\theta}):=G_{2}(\delta).
\een
This finishes (i).

({\it ii}\,) We proceed by contradiction and assume that there is a sequence of data $(\omega^{i},\eta^{i})$ of a sequence of stable horizons $H_{i}$ with $|J_{i}|=1/4$ and having $\lim A_{i}=2\pi$ but not converging to the extreme Kerr horizon $(\omega_{E},\eta_{E})$ (with $|J|=1/4$). From Proposition \ref{PROC} and the discussion after the statement of Theorem \ref{Thmw} we deduce that there is a subsequence of $(\omega^{i},\eta^{i})$ converging in $C^{0}$ to $(\bar{\omega},\bar{\eta})$ of the form     
\begin{align}
\label{E65}\bar{\eta}=\frac{2 (\tan^{2} \frac{\theta}{2})\big/(\tan^{2} \frac{\bar{\theta}}{2})}{1+(\tan^{4} \frac{\theta}{2})\big/(\tan^{4} \frac{\bar{\theta}}{2})
},\qquad\bar{\omega} = \frac{-1+(\tan^{4} \frac{\theta}{2})\big/(\tan^{4} \frac{\bar{\theta}}{2})}{1+(\tan^{4} \frac{\theta}{2})\big/(\tan^{4} \frac{\bar{\theta}}{2})}
\end{align}
where $\bar{\theta}\neq \pi/2$. We will still index such subsequence with ``{i}''.  We will show that this implies that for sufficiently big $i$, the black hole $H_{i}$ is not stable, which is against the assumption. The instability for $i$ big enough is shown by finding a trial function, to be denoted as $\alpha_{\epsilon^{i}}$, for which ${\mathcal S}(h^{i},\alpha_{\epsilon^{i}})<0$, where ${\mathcal S}(h,\alpha)$, for a given metric $h$ and function $\alpha$, is defined to simplify notation here and below as
\ben
{\mathcal S}(h,\alpha)=\int_{H} \big(|\nabla\alpha|^{2}+\kappa \alpha^{2}\big)\, dA.
\een 

We start by noting that the limit metric 
\ben
\bar{h}=\frac{1}{2}e^{\displaystyle -\bar{\sigma}}\, d\theta^{2}+e^{\displaystyle \bar{\sigma}}\sin^{2}\theta\, d\varphi^{2}
\een
where $\bar{\sigma}$ is defined through $\bar{\eta}=e^{\bar{\sigma}}\sin^{2}\theta$ has an angle defect $\beta_{N}$ at the north pole $N$ (i.e. at $\theta=0$) equal to 
\ben
\beta_{N}:=2\pi - \frac{d r}{d l}\bigg|_{l=0}=2\pi (1-\frac{1}{ \tan^{2}\bar{\theta}/2})
\een
and an angle defect at the south pole $S$ (i.e. at $\theta=\pi$) equal to 
\ben
\beta_{S}=2\pi( 1- \tan^{2}\bar{\theta}/2).
\een
Thus, if $\bar{\theta}<\pi/2$ we have $\beta_{N}<0$ and $\beta_{S}>0$, while if instead $\bar{\theta}>\pi/2$ then we have $\beta_{N}>0$ and $\beta_{S}<0$. Assume without loss of generality that $\bar{\theta}<\pi/2$ and therefore that $\beta_{N}<0$. This will be used crucially later.

Denote by $\bar{L}$ the $\bar{h}$-length of the meridian. Define now a function $\alpha(l)$ on $(0,\bar{L}/2]$ by 
\ben
\alpha(l)=\ln\ln \bigg[\frac{e\bar{L}}{2l}\bigg]
\een
and for any $\epsilon<\bar{L}/2$ define $\alpha_{\epsilon}(l)$ on $[0,\infty)$ by 
\ben
\alpha_{\epsilon}(l)=\left\{
\begin{array}{lcl}
\alpha(\epsilon) & \text{when} & l\leq \epsilon,\\
\alpha(l)              & \text{when} & \epsilon \leq l \leq \bar{L}/2,\\
0                           & \text{when} & l\geq \bar{L}/2.
\end{array}
\right.
\een

For any smooth metric $h$ with $L\geq \bar{L}/2$ we can compute ${\mathcal S}(h,\alpha_{\epsilon})$ in the form
\be\label{PSS}
{\mathcal S}(h,\alpha_{\epsilon})=2\pi(1-\phi'(\epsilon))\,\alpha^{2}(\epsilon)+2\pi \int_{\epsilon}^{\bar{L}/2} (\phi\, \alpha'^{2} - \phi''\alpha^{2})\, dl
\ee
where $\phi$ comes from writing $h=dl^{2}+\phi^{2}(l)d\varphi^{2}$, that is $\phi=\sqrt{\eta}$. Note that the limit metric $\bar{h}=dl^{2}+\bar{\phi}^{2}d\varphi^{2}$ is not smooth at the poles and therefore the functional value ${\mathcal S}(\bar{h},\alpha_{\epsilon})$ is, a priori, not well defined. Nevertheless as $\bar{\phi}(l)$ is a smooth function on $[0,\bar{L}]$ the right hand side of (\ref{PSS}) makes also perfect sense if we use $\phi=\bar{\phi}$. 

We prove now two fundamental facts:
\begin{enumerate}
\item[\rm (F1).] If $\epsilon^{*}(\leq \bar{L}/2)$ is sufficiently small then
\ben
\underbrace{2\pi(1-\bar{\phi}'(\epsilon^{*}))\,\alpha^{2}(\epsilon^{*})\vphantom{\int_{\epsilon}}}_{\rm (I)} +\underbrace{2\pi \int_{\epsilon^{*}}^{\bar{L}/2} (\bar{\phi}\, \alpha'^{2} - \bar{\phi}''\alpha^{2})\,dl\,}_{\rm (II)} <\, 0
\een
\item[\rm (F2).] For any $\epsilon^{*}(\leq \bar{L}/2)$ there is a sequence $\epsilon^{i}\rightarrow \epsilon^{*}$ such that
\be\label{RHS}
{\mathcal S}(h^{i},\alpha_{\epsilon_{i}})\rightarrow 2\pi(1-\bar{\phi}'(\epsilon^{*}))\alpha^{2}(\epsilon^{*}) +2\pi \int_{\epsilon^{*}}^{\bar{L}/2} (\bar{\phi}\, \alpha'^{2} - \bar{\phi}''\alpha^{2})\, dl.
\ee  
\end{enumerate}
From (F1) and (F2) it will follow that for $i$ big enough there is $\epsilon^{i}$ close to $\epsilon^{*}$ ($\epsilon^{*}$ as in (F1)) such that ${\mathcal S}(h^{i},\alpha_{\epsilon^{i}})<0$. Thus, we will be done with the proof of Theorem \ref{Thmpoint} after proving (F1) and (F2).

{\it Proof of F1}. From the limits
\begin{align*}
& \lim_{\epsilon^{*}\rightarrow 0} 2\pi(1-\phi'(\epsilon^{*}))=\beta_{N}<0,\\
& \lim_{\epsilon^{*}\rightarrow 0} \alpha^{2}(\epsilon^{*})=\bigg(\ln\ln \bigg[\frac{e\bar{L}}{2\epsilon^{*}}\bigg]\bigg)^{2}\rightarrow +\infty.
\end{align*}
it is deduced that the under-braced term (I) diverges to minus infinity as $\epsilon^{*}$ tends to zero. Hence, to prove (F1) it is enough to prove a bound for $|$(II)$|$ independent of $\epsilon^{*}$. To show this we use the following expansion that the reader can check directly from (\ref{E65}) (recall $\bar{\phi}=\bar{\eta}$) 
\ben
\bar{\phi}(l)=\frac{1}{\tan ^{2}\bar{\theta}/2} l + \frac{\sqrt{2}}{6\tan^{3}\bar{\theta}/2} l^{3} + O(l^{4}).
\een
From it one easily shows that the function $\bar{\phi}'' \alpha^{2}$ is bounded on $[0,\bar{L}]$ and therefore that $|\int_{\epsilon^{*}}^{\bar{L}/2}\bar{\phi}''\alpha^{2}\, dl|\leq  \int_{0}^{\bar{L}} |\phi''\alpha^{2}|\, dl <\infty$. Also, as $(\alpha')^{2}=1/(l\ln l)^{2}$, we have 
\ben
\int_{\epsilon^{*}}^{\bar{L}/2} \alpha'^{2}\bar{\phi} dl\leq \int_{0}^{\bar{L}} \alpha'^{2}\bar{\phi}\, dl<\infty.
\een
This finishes the proof of (F1).

{\it Proof of F2.} An easy application of Rolle's Theorem shows that, as $\phi_{i}$ converges in $C^{0}$ to $\bar{\phi}$, there is a sequence $\epsilon^{i}\rightarrow \epsilon^{*}$ such that $\phi'_{i}(\epsilon^{i})\rightarrow \bar{\phi}'(\epsilon^{*})$. We use this sequence $\epsilon_{i}$ below. After an integrating by parts we obtain the following expression for ${\mathcal S}(h^{i},\alpha_{\epsilon^{i}})$	  
\ben
{\mathcal S}(h^{i},\alpha_{\epsilon^{i}})=2\pi(1-\phi_{i}')\,\alpha^{2}\big|_{\epsilon^{i}}
+\phi'_{i}\alpha^{2}\big|_{\epsilon^{i}}-2\phi_{i}\alpha\alpha'\big|_{\epsilon^{i}}-\int_{\epsilon^{i}}^{\bar{L}/2}2\phi_{i}(\alpha'^{2}+\alpha\alpha'')dl
+2\pi \int_{\epsilon^{i}}^{\bar{L}/2} \phi_{i}\, \alpha'^{2} dl
\een
where we did not write the evaluations at $\bar{L}/2$ which vanish because $\alpha(\bar{L}/2)=0$.
Now as $\phi'_{i}(\epsilon^{i})\rightarrow \bar{\phi}'(\epsilon^{*})$ and $\phi^{i}$ converges to $\bar{\phi}$ in $C^{0}$ we can take the term by term limit in the previous expression to obtain
\ben
2\pi(1-\bar{\phi}')\,\alpha^{2}\big|_{\epsilon^{*}}
+\bar{\phi}'\alpha^{2}\big|_{\epsilon^{*}}-2\bar{\phi}\alpha\alpha'\big|_{\epsilon^{*}}-\int_{\epsilon^{*}}^{\bar{L}/2}2\bar{\phi}(\alpha'^{2}+\alpha\alpha'')dl
+2\pi \int_{\epsilon^{*}}^{\bar{L}/2} \bar{\phi}\, \alpha'^{2} dl.
\een
Undoing the integration by parts we get the right hand side of (\ref{RHS}), as wished.
\end{proof} 
\subsection{Proof of Theorems \ref{ThmRL} and \ref{ThmRL2}.}
\n \begin{proof}[\bf Proof of Theorem \ref{ThmRL}]  It is enough to prove the theorem when $|J|=1/4$.  Recall that the graph of $(\omega,\eta)$ lies between two arcs cutting the $\eta$-axis at the points $(-\delta +\sqrt{\delta^{2}+4})/2$ and $(\delta+\sqrt{\delta^{2}+4})/2$. It follows that
\be\label{Rest11}
R=\max \big\{2\pi \sqrt{\eta}\big\} \leq 2\pi \sqrt{\frac{\delta+\sqrt{\delta^{2}+4}}{2}}.
\ee
On the other hand when the graph of $(\omega,\eta)$ crosses the $\eta$-axis we have $\eta\geq (\delta+\sqrt{\delta^{2}+4})/2$ and therefore  
\be\label{Rest22}
R=\max \big\{2\pi \sqrt{\eta}\big\}\geq 2\pi \sqrt{\frac{-\delta+\sqrt{\delta^{2}+4}}{2}}.\qedhere
\ee
\end{proof}
\n \begin{proof}[{\bf Proof of Theorem \ref{ThmRL2}}] We assume $|J|=1/4$. To obtain the upper bound in (\ref{Rest2}) use (\ref{Rest11}) and 
\ben
\sqrt{\frac{\delta+\sqrt{\delta^{2}+4}}{2}}\leq \big(\delta^{2}+4\big)^{1/4}=\sqrt{\frac{A}{\pi}}.
\een
To obtain the lower bound instead use (\ref{Rest22}) and
\ben
\sqrt{\frac{-\delta+\sqrt{\delta^{2}+4}}{2}}=\sqrt{ \frac{2}{\delta+\sqrt{\delta^{2}+4}}} \geq \sqrt{\frac{\pi}{A}}.
\een 

The first inequality in (\ref{Rest4}) is just $A\geq 8\pi |J|$. The second is (\ref{(a)}) of Theorem \ref{T5} when $l=L$ and therefore $a=A$. 

The first inequality in (\ref{Rest3}) is a consequence of the obvious $LR\geq A$. To obtain the second inequality use $R\leq \sqrt{4\pi A}$ and $\sqrt{A}\leq \sqrt{2\pi}L$ that we have proved before.
\end{proof}

\bibliographystyle{plain}

\bibliography{Master.bib}

\end{document}